\documentclass{article}
\usepackage[utf8]{inputenc}
\usepackage[margin=1in,top=1in,bottom=1in]{geometry}

\usepackage{style}
\usepackage[titlenumbered,algonl]{algorithm2e}

\usepackage
[bordercolor=orange,backgroundcolor=orange!20,linecolor=orange,textsize=scriptsize]
{todonotes}

\newcounter{mycomment}

\title{Strongly Polynomial Frame Scaling to High Precision}
\author{Daniel Dadush
\and Akshay Ramachandran}

\begin{document}

\maketitle

\begin{abstract}
The frame scaling problem is: given vectors $U := \{u_{1}, ..., u_{n} \} \subseteq \R^{d}$, marginals $c \in \R^{n}_{++}$, and precision $\eps > 0$, find left and right scalings $L \in \R^{d \times d}, r \in \R^n$ such that $(v_1,\dots,v_n) := (Lu_1 r_1,\dots,Lu_nr_n)$ simultaneously satisfies $\sum_{i=1}^n v_i v_i^{\mathsf{T}} = I_d$ and $\|v_{j}\|_{2}^{2} = c_{j}, \forall j \in [n]$, up to error $\eps$. 
This problem has appeared in a variety of fields throughout linear algebra and computer science. 
In this work, we give a strongly polynomial algorithm for frame scaling with $\log(1/\eps)$ convergence. 
This answers a question of Diakonikolas, Tzamos and Kane (STOC 2023), who gave the first strongly polynomial randomized algorithm with $\poly(1/\eps)$ convergence for Forster transformation, the special case $c = \frac{d}{n} 1_{n}$. 
Our algorithm is deterministic, applies for general marginals $c \in \R^{n}_{++}$, and requires $O(n^{3} \log(n/\eps))$ iterations as compared to the $O(n^{5} d^{11}/\eps^{5})$ iterations of DTK.
By lifting the framework of Linial, Samorodnitsky and Wigderson (Combinatorica 2000) for matrix scaling to the frame setting, we are able to simplify both the algorithm and analysis. 
Our main technical contribution is to generalize the potential analysis of LSW  to the frame setting and compute an update step in strongly polynomial time that achieves geometric progress in each iteration. 
In fact, we can adapt our results to give an improved analysis of strongly polynomial matrix scaling, reducing the $O(n^{5} \log(n/\eps))$ iteration bound of LSW to $O(n^{3} \log(n/\eps))$. 
Additionally, we prove a novel bound on the size of approximate frame scaling solutions, involving the condition measure $\bar{\chi}$ studied in the linear programming literature, which may be of independent interest. 
\end{abstract}

\section{Introduction}

In this work, we study the following problem:

\begin{definition} [Frame Scaling] \label{d:introForsterTransform}
    A set of vectors $U := (u_{1}, ..., u_{n}) \in \R^{d \times n}$ is called a \textbf{frame} if the matrix $U$ is of full row rank. Given input frame $U$, marginals $c \in \R^{n}_{++}$, and precision $\eps > 0$, find a {\bf left scaling} $L \in \R^{d \times d}$ and a {\bf right scaling} $r \in \R^{n}$, such that $\{v_{j} := L u_{j} r_{j}\}_{j \in [n]}$ are in $\eps$-approximate $(I_{d},c)$-position:
    \[ \Big\| \sum_{j=1}^{n} v_{j} v_{j}^{\Tr} - I_{d} \Big\|_{F}^{2} + \sum_{j=1}^{n} \Big(\|v_{j}\|_{2}^{2} - c_{j} \Big)^{2} \leq \eps^{2} .   \]
    We say $V:=(v_1,\dots,v_n) \in \R^{d \times n}$ is $(I_{d},c)$-scaled if $\eps=0$ in the above. 
\end{definition}

Scaling to approximate $(I_{d},c)$-position can be thought of as a way to regularize a set of vectors. 
The first error term is small when the vectors are globally balanced in the sense that $\|V^{\Tr} x\|_{2} \approx \|x\|_{2}$ for all directions $x \in \R^{d}$; and the second is small when the vector norms nearly match the desired marginals $c \in \R^{n}_{++}$. 
This problem has appeared in a variety of areas including algebraic geometry \cite{GelfandMatroid}, communication complexity \cite{Forster}, functional analysis \cite{BartheBL}, coding theory and signal processing \cite{FramesBook} \cite{HP}. 
The question of existence of frame scalings has been asked and answered many times in the literature: 

\begin{theorem}[\cite{BartheBL,GelfandMatroid}] \label{t:frameExistence}
    Frame $U \in \R^{d \times n}$ can be scaled to $\eps$-approximate $(I_{d},c)$-position for any $\eps > 0$ iff 
    \[  \forall T \subseteq [n]: \quad \langle c, 1_{T} \rangle \leq \rk(U_{T})  \qquad \text{and} \qquad \langle c, 1_{n} \rangle = d .    \]
    We say $(U,c)$ is feasible if this is the case, and otherwise we say it is infeasible. 
\end{theorem}

Some previous proofs of this existence statement use algebraic \cite{GelfandMatroid} or compactness \cite{Forster} arguments, and therefore are non-constructive. Recently, there has been significant focus on algorithms for computing frame scalings \cite{HardtMoitra}, \cite{SinghVishnoi}, which are are useful for data analysis or signal processing. 
Our work is greatly influenced by two recent directions (scaling and learning theory), and we discuss these connections in more detail below. 

The first is a recent line of work on a class of related problems such as operator scaling \cite{GGOW15}, \cite{BFGOWW} and computing Brascamp-Lieb constants \cite{GGOWBL}. These problems fall into the so-called scaling framework, with connections to algebraic geometry and invariant theory. In particular, there has been a great deal of progress in new tractable algorithms for many problems in this framework (including frame scaling), which exploit the underlying algebraic and geometric structure of scalings (see \cite{Towards} for a detailed exposition).  

The second is the recent application of frame scaling to problems in learning theory. In this context, the special case of $c := \frac{d}{n} \vec{1}_{n}$ is known as the Forster transform, as it was originally studied by Forster \cite{Forster} in the context of communication complexity lower bounds. This has been used as a robust method to regularize a dataset in order to speed up various downstream learning tasks \cite{HardtMoitra}, \cite{PointLocation}, \cite{DKT21}, \cite{DKT}. 

Many of the known algorithms for frame scaling exploit a convex formulation of this problem and then apply standard optimization techniques, such as Ellipsoid, gradient descent, or interior point methods \cite{HardtMoitra}, \cite{IPMforGP}, \cite{AKS}, \cite{PointLocation}. The standard analyses of these off-the-shelf methods as applied in these algorithms have weakly polynomial guarantees, in that the number of iterations depends on the bit complexity of the input.

In very recent work, Diakonikolas, Tzamos and Kane \cite{DKT} gave the first strongly polynomial algorithm for computing Forster transformations, the special case $c = \frac{d}{n} \vec{1}_{n}$. They applied this result to give the first strongly polynomial algorithm for halfspace learning in certain noise models. The running time of their algorithm is polynomial in $(d,n, 1/\eps)$, so they asked whether it is possible to extend this strongly polynomial result to the setting of exponentially small $\eps$. In this work, we answer this question in the affirmative by giving a strongly polynomial algorithm for the frame scaling problem with arbitrary $c \in \R^{n}_{+}$ that has $\log(1/\eps)$ runtime dependence. 

Before presenting our main result, we discuss how our scalings are represented. First note that both the isotropy and the norm condition are easy to satisfy individually simply by normalizing. Explicitly, for any left scaling $L$, there is a simple way to compute a right scaling satisfying the norm condition, namely $r_{j}^{2} := c_{j}/\|L u_{j}\|_{2}^{2}$. Similarly, for any right scaling matrix $R = \diag(r_1,\dots,r_n)$, there is a natural corresponding left scaling $L := (U R^{2} U^{\Tr})^{-1/2}$ satisfying the isotropy condition. 
As it is not possible to take square roots in the strongly polynomial model, in our algorithm we choose to maintain the square of the right scaling $z := (r_1^2,\dots,r_n^2) \in \R^{n}_{++}$ while leaving the left scaling $L := (U Z U^{\Tr})^{-1/2}$ implicit. Our main result follows:

\begin{theorem} [Main Theorem (Informal)]
There is a deterministic strongly polynomial time algorithm which, on input frame $U \in \R^{d \times n}$, marginals $c \in \R^{n}_{++}$, and precision $\eps > 0$, takes $O(n^3 \log(n/\eps))$ iterations each requiring $\poly(n,d)$ operations, and outputs either (1) the squared right scaling $z \in \R^{n}_{++}$ of an $\eps$-approximate frame scaling solution; or (2) a certificate of infeasibility $T \subseteq [n]$ satisfying $\langle c, 1_T \rangle > \rk(U_T)$. 
\end{theorem}

As mentioned above, this answers the open question of~\cite{DKT} regarding the dependence on $\eps$ and also resolves the general marginal case. As added benefits, our algorithm is deterministic and substantially simpler than that of~\cite{DKT}. The main reason for the latter is that~\cite{DKT} works with left scalings directly, which are matrices, whereas we work directly with (squared) right scalings, which are just positive vectors. 
We then follow a procedure similar to the classical work of Linial, Samrodnitsky, and Wigderson \cite{LSW}, which gave a strongly polynomial $\poly(n,m)\log(n/\eps)$-time algorithm for the related, but simpler, \emph{matrix scaling problem}. In particular, we are able to lift the potential function analysis used in~\cite{LSW} to the frame setting and argue geometric decrease in each iteration. 
Even for large $\eps > 0$, our algorithm is somewhat more efficient than that of~\cite{DKT} in terms of iterations, as our work requires $O(n^{3} \log(n/\eps))$ iterations to produce an $\eps$-approximate solution, compared to $O(n^{5} d^{11}/\eps^{3})$ iterations required in \cite{DKT}. Both algorithms require a similar $\poly(n,d)$ amount of work for each iteration, neither of which are claimed to be optimized. The formal complexity of each iteration is left to the main body, along with a detailed comparison with prior work (see \cref{t:mainAlgAnalysis}).

Interestingly, applying our techniques to the simpler matrix scaling problem gives a tighter analysis than the original algorithm \cite{LSW}.  
As a consequence, we show that using a slightly optimized update step, one can decrease the number of iterations needed for matrix scaling from $O(n^5 \log(n/\eps))$ to $O(n^3 \log(n/\eps))$ while maintaining roughly the same work per iteration. For a formal definition of the problem and our new result, see \cref{s:matrixImprovement}.

In the rest of the introduction, we give a high-level overview of the techniques (\cref{ss:techniques}), a detailed comparison between the different algorithms (\cref{ss:Comparison}), and a discussion of subtleties in the definition of strongly polynomial  versus weakly polynomial algorithms (\cref{ss:stronglyPoly}).

\subsection{Techniques} \label{ss:techniques}

In this subsection, we give technical overview of our improved algorithm for frame scaling. 
Our high-level strategy is essentially the same as that of \cite{LSW} for matrix scaling.
We iteratively update the scaling so that the norm square of the error decreases sufficiently. Each update step is a uniform scaling of a subset of columns, and the goal is to find a step-size making sufficient progress. How precisely to choose the update direction and step-size is the main technical work. The algorithm of \cite{DKT} for Forster transform is similar in spirit, but the update step and analysis both involve eigenspaces and so are more complicated; we defer more detailed comparison to \cref{ss:Comparison}. 

We begin by helpful simplification: to find an $\eps$-approximate scaling according to \cref{d:introForsterTransform}, it is enough to compute $z \in \R^{n}_{++}$ such that $\|\lev^{U}(z) - c\|_{2}^{2} \leq \eps^{2}$ for 
\[  \lev^{U}_{j}(z) := z_{j} u_{j}^{\Tr} (U Z U^{\Tr})^{-1} u_{j} .    \]
This is equivalent to original frame scaling problem: applying the scalings $L := (U Z U^{\Tr})^{-1/2}, R := \sqrt{Z}$, we get that $V := LUR$ automatically satisfies the isotropy condition $V V^{\Tr} = I_{d}$ and has squared column norms $\{\|v_{j}\|_{2}^{2} = \lev_{j}^{U}(z)\}_{j \in [n]}$. 
The quantity $\lev^U_j(z)$, $j \in [n]$, is the \emph{leverage score} of the $j$th row of $\sqrt{Z} U^\Tr$, where $Z := {\rm diag}(z_1,\dots,z_n)$, an extremely well studied quantity in numerical linear algebra which indicates the ``importance'' of a row. Note that $\lev^{U}(z)$ can be computed in strongly polynomial time as it involves just matrix inversions and multiplications. 
This is a key simplification as we only have to maintain a vector $z \in \R^{n}_{++}$ instead of a matrix. This is the first reason that we choose to work with the squares of right scalings and leave the left scaling implicit. 

Now we can focus on the error $\lev^{U}(z) - c$, which is just a vector. In order to decrease $\|\lev^{U}(z) - c\|_{2}^{2}$ in each iteration, we follow the combinatorial approach of \cite{LSW} and uniformly scale up a subset of columns with large ``margin". Explicitly, we sort the error vector $\lev_{1}^{U}(z) - c_{1} \leq ... \leq \lev_{n}^{U}(z) - c_{n}$, and choose $T = [k]$ to be the prefix set maximizing margin $2\gamma := (\lev_{k+1}^{U}(z) - c_{k+1}) - (\lev_{k}^{U}(z) - c_{k})$. In \cref{p:outlinePolyGap}, we show this maximum margin satisfies $\gamma^{2} \geq \|\lev^{U}(z) - c\|_{2}^{2}/(2n^3)$. From here, our main goal will be to compute a new iterate $z' \in \R^n_{++}$ satisfying $\|\lev^U(z)-c\|_2^2 - \|\lev^U(z')-c\|_2^2 \lesssim \gamma^2$, which decreases the squared error by a $(1-\Omega(1/n^3))$ multiplicative factor, and in turn gives our $O(n^3 \log (n/\epsilon))$ iteration bound.

The next iterate is of the form $z' \gets z \circ (1_{\overline{T}} + \alpha 1_{T})$ for some $\alpha \geq 1$, where $\circ$ denotes the entry-wise product of vectors and $1_T \in \{0,1\}^n$ is the indicator of $T$. It is straightforward to show that this scaling increases all $\lev|_{T}$ and decreases all $\lev|_{\overline{T}}$, the leverages scores indexed by $T$ and $\bar{T} := [n] \setminus T$ respectively. The key of our potential analysis is \cref{l:outlineProgressLemma}, where we show that the decrease in error $\|\lev^{U}(z) - c\|_{2}^{2}$ can be lower bounded in terms of the total amount that the leverage scores increase and the initial margin $\gamma$. 
Formally, we define the proxy function $h(\alpha)$ to be the sum of leverage scores $\lev|_{T}$ when scaling up $z \to z \circ (1_{\overline{T}} + \alpha 1_{T})$. And in \cref{l:outlineProgressLemma}, we show that the error $\|\lev^{U}(z) - c\|_{2}^{2}$ decreases additively by $2\gamma (h(\alpha) - h(1))$ as long as $h(\alpha) - h(1) \leq \gamma$. 

In words, we show that as long as there is still a gap between $\lev|_{T}$ and $\lev|_{\overline{T}}$, then the error $\|\lev^{U}(z) - c\|_{2}^{2}$ is decreasing at a rate proportional to margin $\gamma$ and the difference in leverage scores $\lev|_{T}$. 

If we could find an update $\alpha$ such that $\gamma/5 \leq h(\alpha) - h(1) \leq \gamma$, then this analysis would show geometric progress in each iteration, since we have already argued $\gamma^{2}$ is a polynomial fraction of the potential $\|\lev^{U}(z) - c\|_{2}^{2}$. 
Following a similar argument as \cite{LSW}, we can show that such an update always exists in the feasible case. Specifically, in \cref{p:outlineInfeasible} we show $\lim_{\alpha \to \infty} h(\alpha) = \rk(U_{T})$, whereas the margin $\gamma$ of $T$ effectively shows $h(1) \leq \langle c, 1_{T} \rangle - \gamma$. Using feasibility $\langle c, 1_{T} \rangle \leq \rk(U_{T})$ and the fact that $h$ is increasing, we must eventually have an increase $h(\alpha) - h(1) \geq \gamma/2$. 
In the contrapositive, this shows that a good update does not exist only when the subset $T$ gives a certificate of infeasibility, $\langle c, 1_{T} \rangle > \rk(U_{T})$ according to \cref{t:frameExistence}. 

The rest of the algorithm, actually computing an update in the feasible case, is the main technical contribution of this work. 
Formally, we want an update satisfying $\gamma/5 \leq h(\alpha) - h(1) \leq \gamma$, which is essentially an approximate root-finding problem. 
In our setting, the function $h$ is non-increasing and concave, as we show in \cref{f:HIncreasingConcave}, so we are able to use the classical Newton-Dinkelbach method, along with a refined Bregman divergence-based analysis similar to \cite{AcceleratedND}, in order to efficiently compute the update.

We note that the concavity of $h(\alpha)$ is another helpful property that is a consequence of our choice of representation as squares of right scalings $z \in \R^{n}_{++}$.  

The running time of the Newton-Dinkelbach method applied to our function $h$ is proportional to $O(\log(1+\alpha_*/\alpha_0))$, where $h(\alpha_{*}) = h(1)+\gamma/5$ and $\alpha_0$ is our initial guess. We emphasize that this is a special property of our function $h$ and our desire to just compute an update making sufficient progress; in general, the analysis is more complicated (see \cite{AcceleratedND}). 
 
We rewrite the proxy function as 
\[ h(\alpha) = \tr[\alpha U_{T} Z_{T} U_{T}^{\Tr} (U Z U^{\Tr} + (\alpha-1) U_{T} Z_{T} U_{T}^{\Tr})^{-1}] = \sum_{i=1}^{d} \frac{\alpha \mu_{i}}{1 + (\alpha-1) \mu_{i}} ,   \]
where $1 \geq \mu_1 \geq \cdots \geq \mu_d \geq 0$ are the eigenvalues of $U_{T} Z_{T} U_{T}^{\Tr} (U Z U^{\Tr})^{-1}$. The eigenvalues are in the range $[0,1]$ as the matrix is similar to a principal submatrix of a projection matrix, and the number of non-zero eigenvalues is precisely $\rk(U_T)$. Note that we do not have access to $\mu$, but only to evaluations of $h(\alpha)$ and $h'(\alpha)$. In \cref{l:updateRange}, we are able to use this implicit information to show that the solution lies in a small multiplicative range, either close to $1$ or close to $\gamma/\mu_{s}$ where $\mu_{s} := \sum_{i :\mu_{i} < 1/2} \mu_{i}$ is the sum of small eigenvalues. 

Unfortunately, we cannot compute $\mu_{s}$ in strongly polynomial time, but  we only need a polynomial approximation of this quantity in order to run Newton-Dinkelbach efficiently.
By some preprocessing, we can in fact reduce to the setting where there is a large gap between the small and large eigenvalues. 
In order to compute the spectral sum $\mu_{s}$, our goal is to approximate the subspace of large eigenvectors.
Under the gapped assumption, all the vectors are close to the large eigenspace, so we can approximate this subspace by the span of a column subset. To find a good column subset, we apply a classical combinatorial algorithm of Knuth \cite{Knuth} which finds a subset with large determinant, and adapt the analysis to show that this gives a sufficiently accurate (but very weak) eigendecomposition in our setting. 
Using the resulting initial guess with the Newton-Dinkelbach method allows us to efficiently compute the desired update. 

The final piece of the algorithm is a regularization step that we use to maintain bounded bit complexity. A similar idea was also used in \cite{DKT} for the strongly polynomial Forster transform. But once again, the fact that we only maintain the square of the right scaling $z \in \R^{n}_{++}$ instead of a matrix greatly simplifies our regularization procedure while allowing us to give a more refined bound. 
Formally, in \cref{t:Regularization}, we show that we can always reduce the condition number $\log \frac{z_{\max}}{z_{\min}}$ by shrinking the (multiplicative) gaps between the entries of $z$ while maintaining small error. In fact, we are able to relate the required magnitude and bit complexity of the scaling to the $\bar{\chi}$ condition measure of the frame $U$. This is an extremely well-studied concept in the linear programming literature (see e.g. \cite{ScalingInvIPM}), and importantly is a much more refined condition measure than just the bit complexity of the frame. We believe this relation and our regularization bounds are of independent interest.

\subsection{Relation to Previous Work} \label{ss:Comparison}

In this subsection, we compare our algorithm and analysis to the previous strongly polynomial scaling algorithms of \cite{LSW} and \cite{DKT}. As mentioned in \cref{ss:techniques}, the high-level iterative approach is similar, so we focus on the differences in implementation and analysis that lead to our improvements. 

We first give a formal definition of matrix scaling and compare it to frame scaling. Here, we are given a non-negative matrix $A \in \R^{m \times n}_{+}$, desired row and column sums $(r,c) \in \R^{m+n}_{++}$ and a tolerance $\eps > 0$, and the goal is to compute positive diagonal matrices $L \in \diag_{++}(m), R \in \diag_{++}(n)$ such that the row and column sums $(RA^{\Tr}L 1_m,LAR 1_n)$ of the rescaled matrix $LAR$ are within distance $\eps$ of $(r,c)$. To make the analogy to frame scaling clearer, for input $U \in \R^{d \times n}$, the column sum condition corresponds to the column norm condition on $U$, and the row sum condition corresponds to the isotropy condition. 
We will compare our results with the classical work of \cite{LSW}, which gave a strongly polynomial $\poly(n,m)\log(n/\eps)$-time algorithm for matrix scaling.

We recall the main components of the algorithm: 
\begin{enumerate}
\item Reduce from finding a left and right scaling $L \in \R^{d \times d}$ and $r \in \R^{n}_{++}$, to just the square of a right scaling $z = (r_{1}^{2}, ..., r_{n}^{2}) \in \R^{n}_{++}$. 
\item Use norm squared of the error $\|c-\lev^U(z)\|_2^2$ as our potential function, and attempt to decrease it by a $(1-O(1/n^3))$ multiplicative factor.
\item In each iteration, choose a subset $T \subseteq [n]$ with large ``margin" $\gamma > 0$ to scale up uniformly by a factor $\alpha \geq 1$. That is, $z \rightarrow z \circ (1_{\overline{T}} + \alpha 1_T)$. 
\item Relate the improvement to the proxy function $h_T^U(\alpha)$, the sum of the leverage scores inside $T$, and find a step-size $\alpha \geq 1$ satisfying $h_T(\alpha) \in [\gamma/5,\gamma]$ to get the desired geometric decrease in potential. 
\item Perform a regularization step to keep multiplicative range $\max_i z_i / \min_i z_i$ of the scaling uniformly bounded. This enables us to maintain polynomial bit complexity. 
\end{enumerate}

In the remainder, we will discuss each of the above steps in more detail, specifically how they compare across the different strongly polynomial algorithms, and how our new techniques give improved convergence. 

    \textbf{Representation of Scalings}: All three algorithms reduce to just one-sided scaling. In the matrix setting, for each right scaling $R \in \diag_{++}(n)$, there is a unique left scaling normalizing the rows, namely $L_{ii} := r_{i} / (A R 1_{n})_{i}$. 
    Note that we could equivalently work with column normalized matrices by transposing the input. 
    Therefore, the algorithm of \cite{LSW} maintains a right scaling while implicitly choosing the left scaling that matches the row condition. 

    In the frame setting, there are two inequivalent ways to perform this reduction: 
    any right scaling $R \in \diag_{++}(n)$ induces a unique isotropic left scaling $L := (U R^{2} U^{\Tr})^{-1/2}$; 
    and similarly, any left scaling $L \in \R^{d \times d}$ induces a unique right scaling satisfying the norm conditions, $r_{j}^{2} := c_{j}/\|L u_{j}\|_{2}^{2}$. 
    Our algorithm uses the first approach, whereas the algorithm of \cite{DKT} uses the second. 
    This ends up being a key reason that our algorithm and analysis is quite a bit simpler. As an example, our representation is a positive vector, whereas the algorithm of \cite{DKT} keeps track of the left scaling matrix. And the spectral properties of the left scaling become crucial in their algorithm (e.g. condition number and singular vectors), which require a much more complicated set of procedures, such as a strongly polynomial approximate eigendecomposition. 

    \textbf{Potential Function}: For the potential function, both our work and \cite{LSW} measure the norm square of the error of the column marginals, whereas the potential function in \cite{DKT} measures the error of the isotropy condition $V V^{\Tr} - I_{d}$, which is a matrix quantity. This difference is another source of difficulty for the Forster transform algorithm of \cite{DKT}, as updates no longer just increase or decrease norms, but affect the whole error matrix, and in particular its eigenvalues, in more complicated ways. 

    \textbf{Update Direction}: 
    Both our work and \cite{LSW} choose to uniformly scale up a set of columns with large margin as defined in \cref{ss:techniques}. 

    On the other hand, the potential function of \cite{DKT} is a matrix quantity, and therefore the relevant notion of ``gap" is with respect to its eigenvalues. Assuming an exact eigendecomposition method, the analogous update step would be to scale up the eigenspace corresponding to the small errors. While an exact eigendecomposition cannot be performed in strongly polynomial time, they developed a randomized power method based algorithm to compute a very accurate approximation of the small eigenspace. 
    
    \textbf{Update Step-Size}: 
    This is the simplest in the matrix scaling setting of \cite{LSW}, as the step-size can be computed explicitly in terms of the input matrix and current scalings. In \cref{s:matrixImprovement}, we show how to adapt our proxy function analysis to give an improved update step for matrix scaling that makes quadratically more progress in each step and reduces the number of iterations by the same factor. 
    
    In \cite{DKT}, the update is quite complicated for a few reasons: 
    because the error is a matrix quantity, it is harder to control how scaling up an eigenspace affects the error $\|V V^{\Tr} - I_{d}\|_{F}^{2}$. Specifically, there is an off-diagonal term in Lemma 3.3 of \cite{DKT} which detracts from the progress for large step sizes. 
    For this reason, the algorithm in \cite{DKT} breaks into two cases depending on how the off-diagonal term grows in terms of the scaling. In both cases, the algorithm chooses a fixed step-size that is guaranteed to make $\poly(\eps/nd)$ progress. This is one of the two main bottlenecks for the $\poly(1/\eps)$ runtime, as it is difficult to make geometric progress in the presence of the off-diagonal term. The other main bottleneck is the requirement of a strongly polynomial approximate eigendecomposition procedure. 

    In our work, we relate the error to a proxy function $h_{T}(\alpha)$, which is the sum of leverage scores in $T$ after scaling up by $\alpha \geq 1$. Then we use spectral properties of the column subset $V_{T}$ to compute the appropriate update step making geometric progress. Importantly, we only require a $\poly(n,d)$ factor approximation of the sum of the small eigen values (i.e., of size at most $1/2$), which corresponds to a very coarse eigendecomposition. We are therefore able to use a simple deterministic method to compute this simpler spectral quantity, as opposed to the full eigendecomposition required in \cite{DKT}. Further, our potential function analysis is tighter and allows us to make geometric progress in each iteration.  

    \textbf{Controlling Bit Size}: 
    The algorithm in \cite{LSW} does not use a regularization step, and instead relies on the fact that the magnitude of the scalings at most squares in each iteration. While squaring can double the bit length, the algorithm only requires the top $O(\log(n/\eps))$ most significant bits of each number (scaling or matrix entry) when computing the updates. This is because the algorithm only requires $(1+ \poly(\eps/nd))$ multiplicative approximations of partial row and column sums, and all of these quantities consist only of non-negative numbers. This allows them to use a floating point representation of the entries and scalings, keeping track of only the high order bits, and proving that the exponent of the floating point representation grows by only a polynomial factor throughout the algorithm. 
    
    On the other hand, the frame scaling setting requires performing matrix operations (pseudo-inverses, determinants, projections, rank computations) up to varying levels of accuracy, which are much more sensitive to small errors than approximating sums of non-negative numbers as in matrix scaling. Therefore, both our algorithm and that of \cite{DKT} involve an additional regularization step at the end to control the magnitude and bit complexity of scalings. The regularization procedure in \cite{DKT} is an iterative method involving approximate eigendecompositions and projections of point sets to subspaces. In our setting, we benefit from the fact that our scaling is just a vector of numbers, whereas the left scaling maintained in \cite{DKT} is a matrix. Therefore our regularization procedure just involves arithmetic and comparison operations on the right scaling $z \in \R^{n}_{++}$. 
    There are many previous results for bit complexity bounds for frame scaling in particular (see \cite{SinghVishnoi}, \cite{StraszakVishnoi}, \cite{IPMforGP}), and we use very similar techniques to achieve our bound. In particular, our procedure can be seen as an explicit algorithmic version of the results of \cite{StraszakVishnoi} and \cite{IPMforGP}, which give an implicit algorithm for a more general setting. 
    Furthermore, we are able to relate the required bit complexity to the $\bar{\chi}$ condition measure of the input, which gives a much more refined bound than just bit size (see \cref{s:regularization} for details). 

\subsection{Strong and Weak Polynomial Time} \label{ss:stronglyPoly}

The main result in this work is a faster strongly polynomial algorithm for frame scaling. In this subsection, we will more precisely define the computational model used in this work, as well as some subtleties that require clarification. 

The ``genuinely polynomial" model was first described by Meggido \cite{Megiddo}:
for input given as $n$ real numbers, a strongly polynomial algorithm can perform $\poly(n)$ elementary arithmetic and comparison operations; further, if the input has bit complexity $b$, then the algorithm must have polynomial space complexity, in the sense that intermediate quantities must remain $\poly(n,b)$ bit complexity. 
Subsequently, Smale introduced the real model of computation, where the machine has RAM access to real numbers and an oracle for performing exact arithmetic and comparison operations on these numbers. So a stricter definition is that a strongly polynomial algorithm should work in the real model while magically maintaining polynomial space complexity if the input is given in bounded bit representation. Note that in this interpretation, the algorithm does not have knowledge of the bit representation of the input.

    This work studies strongly polynomial algorithms for $\eps$-approximate frame scaling with $\log(1/\eps)$ convergence.
    Formally, our algorithm performs $\poly(n, \log(1/\eps))$ exact arithmetic and comparison operations in the real model. In order to maintain polynomial space complexity, we require a regularization procedure that rounds intermediate quantities to some precision depending on the bit complexity of the input frame. 
    Such a `rounding' oracle is not present in the strictest definition of strongly polynomial, but is present in slightly weaker variants that have been defined in \cite{GLS}. Further, such an operation is implicitly required in many other strongly polynomial algorithms in the literature. We give a more detailed discussion of this model in \cref{s:stronglyPoly}. 

\subsection{Notation}
We use $\R$, $\R_+$ and $\R_{++}$ to denote the reals, non-negative reals and positive reals respectively.
For $z \in \R^{n}$ we use $Z := \diag(z) \in \R^{n \times n}$. We use $\diag(n)$ and $\diag_{++}(n)$ to denote $n \times n$ diagonal matrices with real and positive real diagonal entries respectively. 
For vectors $\{u_{1}, ..., u_{n}\}$, we will sometimes abuse notation and let $U$ be the matrix with columns $(u_{1}, ..., u_{n})$. 
For univariate function $f(x)$ we use $\frac{d}{dx}$ to be the derivative with respect to $x$.
We use $\circ$ for the entry-wise (Hadamard) product of vectors. That is, $(x \circ y)_i = x_i y_i$ for $x,y \in \R^n$.
For set $T \subseteq [n]$ we use $1_{T}$ to be the indicator vector of a set, and $1_{n}$ to be the all-ones vector. 

\subsection{Organization}

The remainder of this work is structured as follows: in \cref{s:HighLevelAnalysis}, we present our main iterative algorithm and high level convergence analysis for the frame scaling problem (\cref{d:introForsterTransform}). In \cref{s:Guess}, we present and analyze a coarse eigendecomposition procedure that is used as a subroutine within each update iteration. In \cref{s:regularization} we describe and analyze the procedure to round scalings in order to maintain bounded bit complexity. In \cref{s:FinalProofs} we present the convergence analysis of our main algorithm using the guarantees of the previous sections. 
In \cref{s:matrixImprovement} we show how our new techniques give an improved result for strongly polynomial matrix scaling. In \cref{s:stronglyPoly} we give a more detailed discussion of the strongly polynomial model and how it relates to the algorithm presented here. 
Finally in \ref{s:learningTheory}, we discuss how our frame scaling algorithm fits into the known applications for halfspace learning. 

\section{Iterative Algorithm and Analysis} \label{s:HighLevelAnalysis}

In this section, we describe the main iterative algorithm used to solve the frame scaling problem with $\log(1/\eps)$ convergence. The update in each iteration scales up a particular subset uniformly in order to make progress. In the rest of this section, we describe how to choose the subset and step-size. Throughout this paper, we assume for simplicity that the input frames $U \in \R^{d \times n}$ are full row rank, though this assumption can be avoided by considering pseudo-inverses instead of inverses and replacing the factor $d$ with $\rk(U)$ in the analysis. 

We first observe that we can just keep track of the right scaling, as any fixed $r \in \R^{n}_{++}$ induces a unique $L := (U R^{2} U^{\Tr})^{-1/2}$ such that $V := L U R$ satisfies the isotropy condition $V V^{\Tr} = I_{d}$. 
Therefore, in the rest of the work, we will leave the left scaling implicit. For technical reasons, we will keep track of the square of the right scaling $z \in \R^{n}_{++}$, so the implicit left scaling is $L := (U Z U^{\Tr})^{-1/2}$. 
As we show in \ref{s:learningTheory}, this implicit representation is sufficient for downstream applications of frame scaling. 
With this reduction, we can focus on matching the norms of the columns. 



\begin{definition} \label{d:outlineLeverageScores}
    For frame $U \in \R^{d \times n}$, applying right scaling $z \in \R^{n}_{++}$ and transforming to isotropic position gives norms
    \[ \lev_{j}^{U}(z) := z_{j} u_{j}^{\Tr} (U Z U^{\Tr})^{-1} u_{j} = \|(U Z U^{\Tr})^{-1/2} u_{j} \sqrt{z_{j}} \|_{2}^{2} .  \]
    We let $\lev^{U}(z) \in \R^{n}_{+}$ denote the vector of norms.
\end{definition}

    These induced norms are known as leverage scores corresponding to the rows of the matrix $\sqrt{Z} U^{\Tr}$ and are a well-studied linear algebraic parameter with many applications. 
    Note that even though our left/right scalings involve square roots, the norms $\lev^{U}(z)$ can be computed exactly using matrix multiplication and inversion, both of which can be performed in strongly polynomial time \cite{Edmonds}. 

The following simple property of leverage scores will be useful throughout our analysis. 

\begin{fact} \label{f:outlineSumLevScores}
    For any frame $U \in \R^{d \times n}$ and scaling $z \in \R^{n}_{++}$,
    \[ \langle \lev^{U}(z), 1_{n} \rangle = \sum_{j \in [n]} z_{j} u_{j}^{\Tr} (U Z U^{\Tr})^{-1} u_{j} = \tr[ U Z U^{\Tr} (U Z U^{\Tr})^{-1}] = \tr[I_{d}] = d.  \]
\end{fact}

This implies $\langle c, 1_{n} \rangle = d$ is a necessary condition for feasibility of $c \in \R^{n}_{++}$ according to \cref{t:frameExistence}, and we will assume this in the remainder. 
Therefore, by restricting to the isotropic setting, the frame scaling problem simplifies to: 

\begin{definition}  [Isotropic Frame Scaling Problem] \label{d:outlineIsoFrameScalingProblem}
    Given frame $U \in \R^{d \times n}$, marginals $c \in \R^{n}_{++}$ with $\langle c, 1_{n} \rangle = d$, and $\eps > 0$, output either (1) scaling $z \in \R^{n}_{++}$ such that $\|\lev^{U}(z) - c\|_{2}^{2} \leq \eps^{2}$; or (2) certificate of infeasibility: $T \subseteq [n]$ such that $\langle c, 1_{T} \rangle > \rk(U_{T})$. 
\end{definition}

We now present our main algorithm and convergence result. 

\begin{algorithm} \label{a:MainAlgorithm}
\caption{$\operatorname{Main}(U,c,\eps)$}

\KwIn{ Full row rank frame $U \in \R^{d \times n}$, marginals $c \in \R^{n}_{++}, \langle c, 1_{n} \rangle = d$, precision $\eps > 0$ }
\KwOut{ Scaling $z \in \R^{n}_{++}$ s.t. $\|\lev^{U}(z) - c\|_{2}^{2} \leq \eps^{2}$, or certificate of infeasibility: $\langle c, 1_{T} \rangle > \rk(U_{T})$ }
\While{ $\|\lev^{U}(z^{(t)}) - c\|_{2}^{2} > \eps^{2}$ } {
Sort $\lev^{U}_{1}(z^{(t)}) - c_{1} \leq ... \leq \lev^{U}_{n}(z^{(t)}) - c_{n}$\;
Let $k := \arg\max_{j \in [n]} \lev^{U}_{j+1}(z^{(t)}) - \lev^{U}_{j}(z^{(t)})$ be the index with largest gap\;
Choose $T = [k]$ with margin $\gamma := (\lev^{U}_{k+1}(z^{(t)}) - \lev^{U}_{k}(z^{(t)}))/2$ \;
\If{$\rk(U_{T}) < \langle c, 1_{T} \rangle$}{ 
\text{ Output Infeasible\; }
}
\Else{
$\hat{\alpha} \gets \text{SubsetScaleUp}(U,z^{(t)},T,\gamma)$ (Algorithm \ref{a:UpdateAlg})\;
$z^{(t+1)} \gets z^{(t)} \circ (1_{\overline{T}} + \hat{\alpha} 1_{T})$\;
$z^{(t+1)} \gets \text{Regularize}(z^{(t+1)}, \gamma/\poly(n,d))$ (\cref{t:Regularization})\;
$t \gets t+1$\;
}

 }
 Output $z^{(t)}$\;
\end{algorithm}

\begin{theorem} \label{t:mainAlgAnalysis}
    Algorithm \ref{a:MainAlgorithm} either finds a certificate of infeasibility $\langle c, 1_{T} \rangle > \rk(U_{T})$ or produces a scaling $z \in \R^{n}_{++}$ satisfying $\|\lev^{U}(z) - c\|_{2}^{2} \leq \eps^{2}$ in at most $O(n^{3} \log(n/\eps))$ iterations. Each iteration can be implemented using $O(n \log n)$ comparisons, and $O(n d^{2} \log n)$ matrix operations involving $d \times d$ matrices (multiplication, inverse, determinant), and $O(1)$ Gram-Schmidt operations on $d \times n$ matrices. Further, if the entries of $U$ have bit complexity $b$, then all intermediate scalings have bit complexity $O(n ( d b + \log(nd/\eps)))$.  
\end{theorem}

As stated in the introduction, our algorithm requires $O(n^{3} \log(n/\eps))$ iterations, as compared to the $O(n^{5} d^{11}/\eps^{5})$ iteration bound given in \cite{DKT}. We can also compare the runtime of each iteration, and we state this in terms of matrix operations, such as multiplication, inverse, and determinant. In each iteration, our algorithm requires $O(n d^{2} \log n)$ matrix operations on $d \times d$ matrices (the precise runtime of each component of the algorithm is presented in more detail within each relevant section). Each iteration in \cite{DKT} requires at least $O(nd + d^{6} \log(d)/\eps^{2})$ matrix operations on $d \times d$ matrices.

We defer the proof of \cref{t:mainAlgAnalysis} to \cref{s:FinalProofs} after we have collected all components. 
In the remainder of this section, we motivate the steps in Algorithm \ref{a:MainAlgorithm}.

In order to decrease the error $\|\lev^{U}(z) - c\|_{2}^{2}$ while satisfying the strongly polynomial requirement, we choose a combinatorial update step: uniformly scale up the coordinates in a set. The appropriate choice of set is the one maximizing ``margin" defined below. A similar notion appears in both the original matrix scaling algorithm of \cite{LSW} and the Forster transform algorithm of \cite{DKT}. 

\begin{definition} \label{d:SetMargin}
    For input ($U \in \R^{d \times n}, z \in \R^{n}_{++}, c \in \R^{n}_{+}$) the margin of $T \subseteq [n]$ is defined as the largest $\gamma \geq 0$ such that there exists a threshold $\nu \in \R$ satisfying 
    \[ \max_{j \in T} (\lev_{j}^{U}(z) - c_{j}) \leq \nu - \gamma \leq \nu + \gamma \leq \min_{j \not\in T} (\lev_{j}^{U}(z) - c_{j}) .    \]
\end{definition}

We note that $\lev, c \in [0,1]^{n}$, so $\lev - c \in [-1,1]^{n}$ and therefore the margin is at most $1$. 
We will relate the decrease in error to the margin of $T$, so $\gamma$ controls how much progress we can make by scaling up $T$. 
   The next result shows that we can compute a set with large margin:

    \begin{prop} \label{p:outlinePolyGap}
For input ($U \in \R^{d \times n}, z \in \R^{n}_{++}, c \in \R^{n}_{+}$), if $\lev^{U}(z) - c$ is in non-decreasing order $T \subseteq [n]$ is the prefix set with largest margin $\gamma$, then 
\[ \gamma^{2} \geq \frac{1}{2 n^{3}} \|\lev^{U}(z) - c\|_{2}^{2} .    \]
\end{prop}
\begin{proof}
    We will use the property $\langle \lev^{U}(z), 1_{n} \rangle = d$ by \cref{f:outlineSumLevScores} and $\langle c, 1_{n} \rangle = d$ by assumption. Letting $x := \lev^{U}(z) - c$ for shorthand, we have
    \[ 2n \|x\|_{2}^{2} = \sum_{i,j \in [n]} (x_{i} - x_{j})^{2} \leq n^{2} \max_{i,j \in [n]} (x_{i} - x_{j})^{2} \leq n^{2} (2 n \gamma)^{2} = 4 n^{4} \gamma^{2} ,     \]
    where the first step uses $\langle x, 1_{n} \rangle = 0$, and the fourth step uses that the max margin is $\gamma$ so the maximum distance between two consecutive values is $\leq 2 \gamma$. Rearranging gives the result. 
\end{proof}

\begin{remark}
    This can be compared with Lemma 4.7 in \cite{LSW}, where they give a very similar argument for the matrix scaling setting. 
    
    On the other hand, the improving step in Algorithm 2 of \cite{DKT} involves increasing and decreasing quantities. Therefore the update step in \cite{DKT} is much more complicated and requires a strongly polynomial and highly accurate approximate eigendecomposition algorithm. 
\end{remark}

    In the following \cref{ss:ReqUpdate} we discuss how we use the margin to analyze the improvement of $\|\lev - c\|_{2}^{2}$ due to uniformly scaling up subset $T$. Then in the remainder of the section, we give a strongly polynomial algorithm to compute an update making sufficient progress.

\subsection{Update Requirement} \label{ss:ReqUpdate}

    In this section we give explicit conditions for our chosen update step to make sufficient progress in decreasing the error $\|\lev - c\|_{2}^{2}$. We also show that these conditions are always satisfied when the input is feasible, and otherwise the set $T$ provides a certificate of infeasibility. 

    Our plan is to show that our uniform scaling of the coordinates $z_{T}$ decreases the margin $\gamma$, which leads to a decrease in the error. We first observe that the uniform scaling has a simple effect on the leverage scores. 

    \begin{fact} \label{f:LevScoreScaleUp}
        $\lev^{U}_{j}(z \circ (1_{\overline{T}} + \alpha 1_{T}))$ is increasing in $\alpha$ for $j \in T$, and decreasing in $\alpha$ for $j \not\in T$. 
    \end{fact}
    \begin{proof}
    We use $M_{T} := U_{T} Z_{T} U_{T}^{\Tr}, M_{\overline{T}} := U_{\overline{T}} Z_{\overline{T}} U_{\overline{T}}^{\Tr}$ for shorthand and calculate the following matrix derivatives, using $\frac{\rm d}{\rm d\alpha}$ for univariate derivative with respect to $\alpha$:

    \begin{align*}
        \frac{\rm d}{\rm d\alpha} ( M_{\overline{T}} + \alpha M_{T})^{-1} & = - ( M_{\overline{T}} + \alpha M_{T})^{-1} M_{T} ( M_{\overline{T}} + \alpha M_{T})^{-1} \preceq 0 , 
        \\ \frac{\rm d}{\rm d\alpha} \alpha ( M_{\overline{T}} + \alpha M_{T})^{-1} 
        & = ( M_{\overline{T}} + \alpha M_{T})^{-1} - \alpha \frac{\rm d}{\rm d\alpha} ( M_{\overline{T}} + \alpha M_{T})^{-1}
        \\ & = ( M_{\overline{T}} + \alpha M_{T})^{-1} ( M_{\overline{T}} + \alpha M_{T} - \alpha M_{T} ) ( M_{\overline{T}} + \alpha M_{T})^{-1} 
        \\ & = ( M_{\overline{T}} + \alpha M_{T})^{-1} M_{\overline{T}} ( M_{\overline{T}} + \alpha M_{T})^{-1} \succeq 0 , 
    \end{align*}
    where we used the fact $\frac{\rm d}{\rm d\alpha} X_{\alpha}^{-1} = - X_{\alpha}^{-1} (\frac{\rm d}{\rm d\alpha} X_{\alpha}) X_{\alpha}^{-1}$, and the inequalities were because $M_{T}, M_{\overline{T}} \succeq 0$. 

    We can now use this to show the required properties of the leverage score: 
    \begin{align*}
        \lev^{U}_{j \not\in T}(z \circ (1_{\overline{T}} + \alpha 1_{T})) & = 
        z_{j} u_{j}^{\Tr} ( U_{\overline{T}} Z_{\overline{T}} U_{\overline{T}}^{\Tr} + \alpha U_{T} Z_{T} U_{T}^{\Tr})^{-1} u_{j} 
        = z_{j} u_{j}^{\Tr} ( M_{\overline{T}} + \alpha M_{T})^{-1} u_{j} , 
        \\ \lev^{U}_{j \in T}(z \circ (1_{\overline{T}} + \alpha 1_{T})) & = 
        \alpha z_{j} u_{j}^{\Tr} ( U_{\overline{T}} Z_{\overline{T}} U_{\overline{T}}^{\Tr} + \alpha U_{T} Z_{T} U_{T}^{\Tr})^{-1} u_{j} 
        = z_{j} u_{j}^{\Tr} \alpha ( M_{\overline{T}} + \alpha M_{T})^{-1} u_{j}  ,
    \end{align*} 
    by \cref{d:outlineLeverageScores} of $\lev$ and the definitions of $M_{T}, M_{\overline{T}}$. Therefore $\lev_{j \not\in T}$ is non-increasing as the derivative is the quadratic form of a negative semi-definite matrix, and $\lev_{j \in T}$ is non-decreasing as the derivative is the quadratic form of a positive semi-definite matrix. 
    \end{proof}

    This implies the margin is decreasing, as all leverage scores less than the threshold are increasing and all leverage scores greater than the threshold are decreasing. We want to use this to show the error is also decreasing, so we define the following natural proxy function as a way to measure progress: 

    \begin{definition} \label{d:ProgressFunction}
    For frame $U \in \R^{d \times n}$ and scaling $z \in \R^{n}_{++}$, the following function is used to measure the effect of uniformly scaling up $T \subseteq [n]$:
    \[ h^{U,z}_{T}(\alpha) := \sum_{j \in T} \lev^{U}_{j}(z \circ (1_{\overline{T}} + \alpha 1_{T})) 
    = \tr[\alpha U_{T} Z_{T} U_{T}^{\Tr} (U_{\overline{T}} Z_{\overline{T}} U_{\overline{T}}^{\Tr} +\alpha U_{T} Z_{T} U_{T}^{\Tr})^{-1}]
       \]
    \end{definition}

    \begin{remark} \label{r:hStronglyPoly}
        For our update procedure, we will need to perform computations involving our proxy function in strongly polynomial time. We note the following explicit expressions: 
        \[ h(\alpha) = \tr[ \alpha M_{T} (M_{\overline{T}} + \alpha M_{T})^{-1} ], \qquad h'(\alpha) = \tr[ M_{T} (M_{\overline{T}} + \alpha M_{T})^{-1} M_{\overline{T}} (M_{\overline{T}} + \alpha M_{T})^{-1}]  ,  \]
        where we use the shorthand $M_{T} = U_{T} Z_{T} U_{T}^{\Tr}, M_{\overline{T}} = U_{\overline{T}} Z_{T} U_{\overline{T}}^{\Tr}$ and calculations given in \cref{f:LevScoreScaleUp}. 
        Therefore both $h, h'$ can be computed using simple matrix multiplication and inversion. 
    \end{remark}

    With these definitions in hand, we prove a structural lemma showing that the update step decreases the error proportional to the margin of the set and the proxy function $h$. This simple lemma is the key to our improved iteration bound.  

\begin{lemma} \label{l:outlineProgressLemma}
    For input ($U \in \R^{d \times n}, z \in \R^{n}_{++}, c \in \R^{n}_{+}$), let $T \subseteq [n]$ have margin $\gamma$ according to \cref{d:SetMargin} and let $h := h^{U,z}_{T}$ be the progress function given in \cref{d:ProgressFunction}. 
    Then for any scaling $z' := z \circ (1_{\overline{T}} + \alpha 1_{T})$ with $\alpha \geq 1$ satisfying $0 \leq h(\alpha) - h(1) \leq \gamma$, 
    \[ \|\lev^{U}(z) - c\|_{2}^{2} - \|\lev^{U}(z') - c\|_{2}^{2} \geq 2 \gamma (h(\alpha) - h(1)) .  \]
\end{lemma}
\begin{proof}
    For shorthand, let $\lev := \lev^{U}(z)$ and $\lev' := \lev_{j}^{U}(z')$. 
    Then
    \begin{align*}
        \|\lev - c\|_{2}^{2} - \|\lev' - c\|_{2}^{2} 
        & = \|\lev - c - \nu 1_{n}\|_{2}^{2} - \|\lev' - c - \nu 1_{n}\|_{2}^{2} 
        = \sum_{j \in [n]} \Big( (\lev_{j} - \lev'_{j})^{2} + 2 (\lev_{j} - \lev'_{j}) (\lev'_{j} - c_{j} - \nu) \Big)
        \\ & = \sum_{j \in T} (\lev'_{j} - \lev_{j}) \Big(2 \nu - (\lev_{j} - c_{j}) - (\lev'_{j} - c_{j}) \Big)
        + \sum_{j \not\in T} (\lev_{j} - \lev'_{j}) \Big((\lev_{j} - c_{j}) + (\lev'_{j} - c_{j}) - 2 \nu \Big)
    \end{align*}
    where in the first step we subtract the threshold difference of norms as $\langle \lev - c, 1_{n} \rangle = \langle \lev' - c, 1_{n} \rangle$ by \cref{f:outlineSumLevScores}.  

    Now we focus on the first sum. Since we are uniformly scaling up $T$ with $\alpha \geq 1$, we have $\lev'_{j} \geq \lev_{j}$ for all $j \in T$ by \cref{f:LevScoreScaleUp}. Further, by definition of margin, we have $\nu - (\lev_{j} - c_{j}) \geq \gamma$ for all $j \in T$.
 
    Finally, our condition on $h := h_{T}^{U,z}$ implies 
    \[ \lev'_{j} - \lev_{j} \leq \sum_{j \in T} (\lev'_{j} - \lev_{j}) = h(\alpha) - h(1) \leq \gamma    ,    \]
    where we used $\lev'_{j} \geq \lev_{j}$ for all $j \in T$ by \cref{f:LevScoreScaleUp} and our assumption $h(\alpha) - h(1) \leq \gamma$. 

    This implies $\nu - (\lev'_{j} - c_{j}) \geq 0$ for all $j \in T$, so we can lower bound 
    \[ \sum_{j \in T} (\lev'_{j} - \lev_{j}) \Big(2 \nu - (\lev_{j} - c_{j}) - (\lev'_{j} - c_{j}) \Big) \geq \sum_{j \in T} (\lev'_{j} - \lev_{j}) \gamma .  \]
    We can use a symmetric argument to lower bound the $j \not\in T$ terms,
    noting $\lev'_{j} \leq \lev_{j}, (\lev_{j} - c_{j}) - \nu \geq \gamma$ in this case, along with the fact that $\langle 1_{\overline{T}}, \lev - \lev' \rangle = \langle 1_{T}, \lev' - \lev \rangle$ by \cref{f:outlineSumLevScores}. 

    Putting these together, we have
    \begin{align*}
        \|\lev - c\|_{2}^{2} - \|\lev' - c\|_{2}^{2} 
        \geq \gamma \bigg( \sum_{j \in T} (\lev_{j}' - \lev_{j}) + \sum_{j \not\in T} (\lev_{j} - \lev'_{j}) \bigg) 
        = 2 \gamma \sum_{j \in T} (\lev_{j}' - \lev_{j}) = 2 \gamma (h(\alpha) - h(1)) , 
    \end{align*}  
    where in the second step we used $\langle \lev' - \lev, 1_{T} \rangle = \langle \lev - \lev', 1_{\overline{T}} \rangle$ and the last step was by \cref{d:ProgressFunction} of $h$. 
\end{proof}

\begin{remark}
    This can be compared with Lemma 4.7 in \cite{LSW}, where they give a lower bound on the progress in terms of the largest entry-wise change. In \cref{t:updateGuarantee}, we show how to compute an update making $\gamma^{2}$ progress by a more refined analysis of the progress function $h$. Applying our analysis in the matrix setting gives an $O(n^{2})$ improvement in the runtime, as discussed in \cref{s:matrixImprovement}.  
    
    The algorithm of \cite{DKT} measures error in terms of the difference of the matrix $V V^{\Tr}$ to the identity. Therefore, their analysis of the update in Lemma 3.3 includes ``off-diagonal" terms which may increase the error. In order to guarantee that the error decreases, their step size is bounded by $\poly(\eps)$, which leads to their $\poly(1/\eps)$ convergence. 
    \end{remark}

    Note that by \cref{p:outlinePolyGap}, we can always find a set with margin at least a polynomial fraction of the error $\|\lev - c\|_{2}^{2}$. Therefore if we could find $\gamma \geq h(\alpha) - h(1) \geq \gamma/\poly(n,d)$ in each iteration, this would give $\log(1/\eps)$ convergence. 
    In \cref{ss:updateRange}, we prove some structural lemmas about the progress function $h$ and use it to give an explicit expression for a good update. Then in \cref{ss:updateCompute}, we show how to compute this update in strongly polynomial time.

\subsection{Update Range} \label{ss:updateRange}

    The goal of this subsection is to give an explicit expression for a good update step which allows us to make geometric progress in each iteration. 

    We first show that a good update always exists in the feasible case.  

\begin{prop} \label{p:outlineInfeasible}
    For progress function $h := h_{T}^{U,z}$ as in \cref{d:ProgressFunction}, 
    \begin{enumerate}
        \item $h$ is an increasing function of $\alpha$;
        \item $\lim_{\alpha \to \infty} h(\alpha) = \rk(U_{T})$;
        \item $\langle c, 1_{T} \rangle - h(1) \geq \gamma$ where $\gamma$ is the margin of set $T$ according to \cref{d:SetMargin};
        \item If $\rk(U_{T}) \geq \langle c, 1_{T} \rangle$, then for margin $\gamma$ and any $\delta \in [0,\gamma)$, there is a finite solution $\alpha < \infty$ to the equation $h(\alpha) = h(1) + \delta$.
    \end{enumerate}
\end{prop}

\begin{proof}

    (1) follows simply from \cref{f:LevScoreScaleUp} as $h$ is the sum of $\lev_{j \in T}$ which are increasing. 
    
    For item (2), we have
    \begin{align*}
    \lim_{\alpha \to \infty} h(\alpha) & = 
    \lim_{\alpha \to \infty}  \tr[ \alpha U_{T} Z_{T} U_{T}^{\Tr}
    (U_{\overline{T}} Z_{\overline{T}} U_{\overline{T}}^{\Tr} + \alpha U_{T} Z_{T} U_{T}^{\Tr})^{-1} ]
    \\ & = \lim_{\lambda \to 0} \tr[ U_{T} Z_{T} U_{T}^{\Tr}, (\lambda U_{\overline{T}} Z_{\overline{T}} U_{\overline{T}}^{\Tr} + U_{T} Z_{T} U_{T}^{\Tr} )^{-1} ] 
    \\ & = \tr[ (U_{T} Z_{T} U_{T}^{\Tr}) (U_{T} Z_{T} U_{T}^{\Tr})^{+} ]  = \rk(U_{T} Z_{T} U_{T}^{\Tr})
    = \rk(U_{T}) , 
    \end{align*}
    where we used change of variable $\lambda := \alpha^{-1}$. 

    In the remainder of the proof, we use $\lev := \lev^{U}(z)$ for shorthand. 
    Because $T$ has margin $\gamma$, it must be the case that either $\max_{j \in T} \lev_{j} - c_{j} \leq - \gamma$ or $\min_{j \not\in T} \lev_{j} - c_{j} \geq \gamma$. By \cref{f:outlineSumLevScores} we have $\langle c - \lev, 1_{T} \rangle = -\langle c - \lev, 1_{\overline{T}} \rangle$, 
    so in either case we have $\langle c, 1_{T} \rangle - h(1) = \langle c - \lev, 1_{T} \rangle \geq \min \{|T|, |\overline{T}|\} \gamma$, which establishes item (3). 

    For the final item,
    we can lower bound the maximum possible progress 
    \[\sup_{\alpha \in [1,\infty]} h(\alpha) - h(1) = \lim_{\alpha \to \infty} h(\alpha) - h(1) = \rk(U_{T}) - \langle \lev, 1_{T} \rangle = (\rk(U_{T}) - \langle c, 1_{T} \rangle) + \langle c - \lev, 1_{T} \rangle \geq \gamma ,  \]
    where the first step was by item (1), the second was by item (2), and in the last step we used $\rk(U_{T}) \geq \langle c, 1_{T} \rangle$ by feasibility and lower bounded the second term using item (3). 

    Therefore for every $\delta \in [0,\gamma)$, there is a finite solution to the equation $h(\alpha) = h(1) + \delta$. 
\end{proof}

\begin{remark}
    This can be compared to Theorem 4.1 in \cite{LSW} on matrix scaling, where they show that for feasible inputs there is always an update decreasing the error by a $(1 - 1/\poly(n))$ multiplicative factor. 

    Similarly, at the end of Proposition 3.8 in \cite{DKT}, the authors show that if the update step does not succeed, then it produces a certificate of infeasibility. 
\end{remark}

Item (4) in the contrapositive implies that the only reason we cannot make sufficient progress $\Omega(\gamma)$ in the proxy function $h$ is when the set $T$ gives a certificate of infeasibility $\langle c, 1_{T} \rangle > \rk(U_{T})$ according to \cref{t:frameExistence}.  
Combining this with the analysis in \cref{l:outlineProgressLemma} shows that in the feasible case we can always make $1/\poly(n)$ multiplicative progress in the error.

We will use structural properties of the proxy function to compute our update, so for this purpose we rewrite $h$ as a sum of simple constituent functions. 

    \begin{lemma} \label{l:outlineHParts}
        For frame $U \in \R^{d \times n}$, scaling $z \in \R^{n}_{++}$, and set $T \subseteq [n]$, let $\mu := {\rm spec}((U Z U^{\Tr})^{-1} U_{T} Z_{T} U_{T}^{\Tr})$.
        The function $h := h_{T}^{U,z}$ given in \cref{d:ProgressFunction} can be rewritten as
        \[ h(\alpha) = \sum_{i=1}^{d} \frac{\alpha \mu_{i}}{1-\mu_{i} + \alpha \mu_{i}}, \qquad 
        h(\alpha) - h(1) = \sum_{i=1}^{d} \frac{(\alpha-1) \mu_{i} (1 - \mu_{i})}{1 + (\alpha-1) \mu_{i}} , 
        \qquad h'(\alpha) = \sum_{i=1}^{d} \frac{\mu_{i} (1-\mu_{i})}{(1 + (\alpha - 1)\mu_{i})^2} . 
        \]
    \end{lemma}
    \begin{proof}

    Let $V := (U Z U^{\Tr})^{-1/2} U Z^{1/2}$, and note that $V_{T} V_{T}^{\Tr}$ is similar to $(U Z U^{\Tr})^{-1} U_{T} Z_{T} U_{T}^{\Tr}$ and so has the same spectrum $\mu$. Therefore we can rewrite
\[ h(\alpha) = \tr[  \alpha U_{T} Z_{T} U_{T}^{\Tr}
    (U_{\overline{T}} Z_{\overline{T}} U_{\overline{T}}^{\Tr} + \alpha U_{T} Z_{T} U_{T}^{\Tr})^{-1} ]
= \alpha \tr[ V_{T} V_{T}^{\Tr} (I_{d} + (\alpha-1) V_{T} V_{T}^{\Tr})^{-1} ] = \sum_{i=1}^{d} \frac{\alpha \mu_{i}}{1 + (\alpha-1) \mu_{i}} , 
\]
    where in the last step we used that $V_{T} V_{T}^{\Tr}$ and $(I_{d} + (\alpha-1) V_{T} V_{T}^{\Tr})^{-1}$ can be simultaneously diagonalized. 

    The remaining expressions are straightforward calculations. 

    \end{proof}

    Note that this gives very simple proofs of certain properties shown above using matrix methods: e.g. \cref{p:outlineInfeasible}(1) follows from the fact that $\alpha \to \frac{(\alpha-1) \mu(1-\mu)}{1 + (\alpha-1) \mu}$ is increasing, and (2) follows from $\lim_{\alpha \to \infty} \frac{(\alpha-1) \mu(1-\mu)}{1 + (\alpha-1) \mu} = 1-\mu$.
    We now state the main result of this subsection. 

    \begin{lemma} \label{l:updateRange}
        Consider feasible input $(U,z,T)$ according to \cref{p:outlineInfeasible} with progress function $h := h_{U,z}^{T}$ and gap $\gamma \in (0,1]$. Assume $\lim_{\alpha \to \infty} h(\alpha) \geq h(1) + \gamma$, and let $\alpha_{*}$ satisfy $h(\alpha_{*}) = h(1) + \gamma/5$. 
    \begin{enumerate}
        \item If $h'(1) \geq \gamma/4$, then $\alpha_{*} \leq 4$. 
        \item Otherwise, if $h'(1) < \gamma/4$, let $\mu_{s} := \sum_{i: \mu_{i} < 1/2}$. Then $\mu_{s} > 0$ and $\alpha_{*} \leq 1 + \frac{\gamma}{\mu_{s}}$. More generally, for any $\alpha := 1 + \delta/\mu_{s}$ with $\delta \in [0,1]$:
        \[ \frac{\delta}{4} \leq h(\alpha) - h(1) \leq \frac{\gamma}{2} + \delta .    \]
    \end{enumerate}
    \end{lemma}

    In the remainder of this subsection, we collect some structural results in order to prove the above lemma. 

    \begin{corollary} \label{f:HIncreasingConcave}
        In the setting of \cref{l:outlineHParts}, the proxy function $h$ is increasing and concave. 
    \end{corollary}
    \begin{proof}
        This follows simply from \cref{l:outlineHParts} by noting $h$ is a sum of functions of the form $\frac{\alpha \mu}{1-\mu + \alpha \mu}$, which can be easily verified to be increasing and concave for $\mu \in [0,1]$. 
    \end{proof}

    Root-finding for increasing concave functions is exactly the setting for the Newton-Dinkelbach method, and we show how to use this method to compute the solution in \cref{ss:updateCompute}. 

    We next establish a structural inequality of our function $h$, showing it is approximately linear in a small neighborhood.  

        \begin{claim} \label{f:polarizedUpdate}
        Consider input $(U,z)$ and $T \subseteq [n]$ with progress function $h := h^{U,z}_{T}$ according to \cref{d:ProgressFunction}. Then for any $1 \leq \alpha \leq \alpha'$, 
        \[  \frac{\alpha}{\alpha'} (\alpha'-\alpha) h'(\alpha) \leq  h(\alpha') - h(\alpha)  \leq (\alpha'-\alpha) h'(\alpha).  \]
        Equivalently, for Bregman divergence $D(\beta | \alpha) := h'(\alpha) (\beta - \alpha) + h(\alpha) - h(\beta)$,
        \[  0 \leq D(\alpha' | \alpha) \leq \bigg( \frac{\alpha'}{\alpha} - 1 \bigg) (h(\alpha') - h(\alpha)) .    \]
    \end{claim}
    \begin{proof}
        We focus on proving the first statement, as the second line is a simple rearrangement. 
        The upper bound is direct from concavity of $h$ according to \cref{f:HIncreasingConcave}
        For the lower bound, we write out the expression for progress as given in \cref{l:outlineHParts}:
        \begin{align*}
            h(\alpha') - h(\alpha) & = \sum_{i=1}^{d} \bigg( \frac{\alpha' \mu_{i}}{1 + (\alpha'-1)\mu_{i}} - \frac{\alpha \mu_{i}}{1 + (\alpha-1)\mu_{i}} \bigg)
            = \sum_{i=1}^{d}  \frac{(\alpha'-\alpha) \mu_{i} (1-\mu_{i})}{(1 + (\alpha'-1)\mu_{i}) (1 + (\alpha-1)\mu_{i}) }
            \\ & = \sum_{i=1}^{d} \frac{(\alpha'-\alpha) \mu_{i} (1-\mu_{i})}{(1 + (\alpha-1)\mu_{i})^{2}} \frac{1 + (\alpha-1)\mu_{i}}{1 + (\alpha'-1) \mu_{i}} 
            \geq \frac{\alpha}{\alpha'}  \sum_{i=1}^d \frac{(\alpha'-\alpha) \mu_{i} (1-\mu_{i})}{(1 + (\alpha-1)\mu_{i})^{2}} = \frac{\alpha}{\alpha'} (\alpha'-\alpha) h'(\alpha)  , 
        \end{align*} 
        where the inequality follows since $\frac{1+(\alpha-1)\mu}{1+(\alpha'-1)\mu} \geq \frac{\alpha}{\alpha'}$ for $1 \leq \alpha \leq \alpha'$ and $\mu \in [0,1]$ and we substitute the expression for the derivative $h'(\alpha)$ in the final step.
    \end{proof}

    We can now prove the explicit expression for our update. 

    \begin{proof} [Proof of \cref{l:updateRange}]
    The case (1) $h'(1) \geq \gamma/4$ follows simply from the previous claim:
    \[ \frac{\gamma}{4} (\alpha_{*} - 1) \leq h'(1) (\alpha_{*} - 1) \leq \alpha_{*} (h(\alpha_{*}) - h(1)) = \alpha_{*} \frac{\gamma}{5}  ,  \]
    where in the first step we used assumption $h'(1) \geq \gamma/4$, the second step was by the lower bound in \cref{f:polarizedUpdate}, and the final step was by definition $h(\alpha_{*}) = h(1) + \gamma/5$. Rearranging gives $\alpha_{*} \leq 4$. 

    For the first part of (2), first assume for contradiction that $\mu_{s} = 0$ so $\mu_{i \leq r} \geq 1/2$ where $r := \rk(U_{T})$. Then this gives
    \[ \frac{\gamma}{4} > h'(1) = \sum_{i=1}^{r} \mu_{i} (1-\mu_{i}) \geq \frac{1}{2} \sum_{i=1}^{r} (1-\mu_{i}) = \frac{1}{2} (\lim_{\alpha \to \infty} h(\alpha) - h(1) ) \geq \frac{\gamma}{2} ,   \]
    where the first step was by assumption $h'(1) < \gamma/4$, in the second step we used that $\mu_{i \leq r} \geq 1/2$, in the third step we applied $\lim_{\alpha \to \infty} h(\alpha) = r$ by \cref{p:outlineInfeasible}(2) and $h(1) = \sum_{i=1}^{r} \mu_{i}$, and the final step was by assumption. This gives the required contradiction, so we must have $\mu_{s} = \sum_{i : \mu_{i} < 1/2} \mu_{i} > 0$. 

    Next, the bound $\alpha_{*} \leq 1 + \gamma/\mu_{s}$ follows from the general statement applied with $\delta = \gamma$:
    \[ h(1 + \gamma/\mu_{s}) \geq h(1) + \gamma/4 \geq h(\alpha_{*}) \implies 1 + \gamma/\mu_{s} \geq \alpha_{*} ,   \]
    where the implication was by monotonicity of $h$ according to \cref{f:HIncreasingConcave}. 
        
        To prove the general claim, we separate into two terms, depending on large and small eigenvalues:
        \[ h(\alpha) - h(1) = \sum_{i=1}^{d} \frac{(\alpha-1) \mu_{i} (1-\mu_{i})}{1 + (\alpha-1) \mu_{i}} 
        = \sum_{i : \mu_{i} \geq 1/2} \frac{(\alpha-1) \mu_{i} (1-\mu_{i})}{1 + (\alpha-1) \mu_{i}} + \sum_{i : \mu_{i} < 1/2} \frac{(\alpha-1) \mu_{i} (1-\mu_{i})}{1 + (\alpha-1) \mu_{i}}   \]
        as given by \cref{l:outlineHParts}. We can bound the large eigenvalue term as
        \[ 0 \leq \sum_{i : \mu_{i} \geq 1/2} \frac{(\alpha-1) \mu_{i} (1-\mu_{i})}{1 + (\alpha-1) \mu_{i}} \leq \sum_{i : \mu_{i} \geq 1/2} (1-\mu_{i}) \leq 2 \sum_{i : \mu_{i} \geq 1/2} \mu_{i} (1-\mu_{i}) \leq 2 h'(1) ,   \]
        where non-negativity is clear, in the second step we used the simple bound $(\alpha-1) \mu \leq 1 + (\alpha-1)\mu$, in the third step we used $\mu_{i} \geq 1/2$, and in the final step we used the expression in \cref{l:outlineHParts} for the derivative. 
        
        For the other term, we substitute $\alpha = 1 + \delta/\mu_{s}$ for $\mu_{s} := \sum_{i : \mu_{i} < 1/2} \mu_{i}$ and upper bound
        \[
            \sum_{i : \mu_{i} < 1/2} \frac{(\alpha-1) \mu_{i} (1-\mu_{i})}{1 + (\alpha-1) \mu_{i}} 
            \leq \frac{\delta}{\mu_{s}} \sum_{i : \mu_{i} < 1/2} \mu_{i} = \delta , \]
        using the simple bound $\frac{1-\mu}{1 + (\alpha-1) \mu} \leq 1$. Similarly, we can lower bound
        \[ 
            \sum_{i : \mu_{i} < 1/2} \frac{(\alpha-1) \mu_{i} (1-\mu_{i})}{1 + (\alpha-1) \mu_{i}} 
            \geq \frac{\delta}{\mu_{s}} \left( \sum_{i : \mu_{i} < 1/2} \mu_{i} \right) \frac{1/2}{1+\delta} \geq \frac{\delta}{4}  \]
        where first step was by $1 - \mu_{i} \geq 1/2$ for the numerator and $\mu_{i}/\mu_s \leq 1$ for the denominator, and the last step was by our assumption $\delta \leq 1$. 
        Putting together the bounds for both terms and using the condition $h'(1) \leq \gamma/4$ gives the statement
         \[ \frac{\delta}{4} \leq h(1 + \delta/\mu_{s}) - h(1) \leq 2 h'(1) + \delta  \leq \frac{\gamma}{2} + \delta .    \]

    \end{proof}

 In the following subsection, we show how to find such an update in strongly polynomial time.

\subsection{Computing the Update} \label{ss:updateCompute}

The goal of this section is to compute an update $\hat{\alpha}$ for which we can apply \cref{l:outlineProgressLemma} to get polynomial progress. 
If we had access to the spectrum $\mu$, then we could compute a good update using the explicit expression given in \cref{l:updateRange}.     
    In our setting, we only have access to the frame $U$, so one approach would be to approximate the eigenvalues. It turns out that we can bypass this procedure and use a more implicit method to compute the solution. 

Below we give the pseudocode and analysis for our update algorithm. 

\begin{theorem} \label{t:updateGuarantee}
    Given input $(U,z,T,\gamma)$ that is feasible according to \cref{p:outlineInfeasible}, Algorithm \ref{a:UpdateAlg} runs in strongly polynomial time and finds $\hat{\alpha} \geq 1$ satisfying 
    \[ \gamma/5 \leq  h(\hat{\alpha}) - h(1) \leq \gamma , \]
    where $h := h_{T}^{U,z}$ according to \cref{d:ProgressFunction}. 
\end{theorem}

\begin{algorithm} \label{a:UpdateAlg}
\caption{$\operatorname{Update}(U,z,T,\gamma)$}

\KwIn{$U \in \R^{d \times n}$, $\rk(U) = d$, scaling $z$, $T \subseteq [n]$, margin $\gamma \in (0,1]$ 
satisfying $\sup_{\alpha \geq 1} h(\alpha) \geq h(1)+\gamma$.}
\KwOut{ $\hat{\alpha} \geq 1$ s.t. $\gamma/5 \leq h(\hat{\alpha}) - h(1) \leq \gamma$ for $h := h_{T}^{U,z}$ (\cref{d:ProgressFunction}) }
\uIf{$h'(1) \geq \gamma/4$}{
    $\alpha_0 \leftarrow 1$\;
}\Else{
   $\tilde{\mu} \leftarrow \operatorname{Approx-Small-Eigen-Sum}(U,z,T)$ (Algorithm \ref{a:GuessAlgorithm})\;
   $\alpha_0 \leftarrow 1 + \frac{\gamma}{2\tilde{\mu}}$\;
}
Output $\text{ND}(h, \alpha_{0}, b' := h(1) + \gamma/4, b := h(1) + \gamma)$ (Algorithm \ref{a:NDAlgorithm})\;
\end{algorithm}

    The key observation is that our function $h$ is increasing and concave by \cref{f:HIncreasingConcave}. Therefore, we can apply the Newton-Dinkelbach method to compute an approximate solution to the equation $h(\hat{\alpha}) - h(1) = \gamma$. The method and its analysis are given below.

\begin{algorithm} \label{a:NDAlgorithm}
\caption{$\operatorname{ND}(f,\alpha_0,b',b)$}

\KwIn{Differentiable and increasing concave function $f : \R \to \R$, lower bound $b'$, upper bound $b$ satisfying $\sup_\alpha f(\alpha) \geq b$ and $b' < b$, starting guess $f(\alpha_{0}) \leq b$.}
\KwOut{$\hat{\alpha}$ satisfying $f(\hat{\alpha}) \in [b',b]$}
$t \leftarrow 0$\;
\While{$f(\alpha_{t}) < b'$} {
    $\alpha_{t+1} := \alpha_{t} + \frac{b - f(\alpha_{t})}{f'(\alpha_{t})}$\;
    $t \gets t+1$\;}
    Output $\alpha_{t}$\;
\end{algorithm}

    The convergence of this method is well-known, and we used a refined iteration bound using Bregman divergence arguments similar to the work in \cite{AcceleratedND}. 

    \begin{theorem} \label{t:outlineNDAnalysis}
    Let $f : \R \to \R$ be a differentiable increasing concave function with $b' < b \in \R$. Let $\alpha_* \in \R$ satisfy $f(\alpha_{*}) = b'$ and let $\alpha_0 \in \R$ satisfy $f(\alpha_0) \leq b$. 
    
    Then the Newton-Dinkelbach method in Algorithm \ref{a:NDAlgorithm} applied with starting guess $\alpha_0$ outputs: $\alpha_{0}$ if $f(\alpha_{0}) \geq b'$; otherwise it outputs $\hat{\alpha}$ satisfying $f(\hat{\alpha}) \in [b',b]$ in at most $T$ iterations, where
    \[ 
       T := 1 + \left\lceil \log_{\frac{1}{1-\eps}} \left(\max \left\{1, \frac{D_{f}(\alpha_{*} |\alpha_{0})}{b-b'}\right\} \right) \right\rceil , 
    \] 
    for $\eps \in (0,1)$ satisfying $\eps(b-f(\alpha_0)) \leq b-b'$,
    and $D_{f}(\beta |\alpha) := f'(\alpha)(\beta -\alpha)+f(\alpha)-f(\beta)$ which is known as the Bregman divergence of $f$. 
 
    \end{theorem}
    \begin{proof}
        The first case $b' \leq f(\alpha_0) \leq b$ is trivial, so we assume $f(\alpha_0) < b'$. Let $\alpha_0 < \alpha_1 < \alpha_2,\dots$ denote the iterates of the Newton Dinkelbach algorithm.
        By concavity of $f$, every iterate $\alpha_t$, $t \geq 1$, satisfies $f(\alpha_t) \leq f'(\alpha_{t-1})(\alpha_t-\alpha_{t-1}) + f(\alpha_{t-1}) = b$. Thus, it suffices to upper bound the index of the first iterate satisfying $f(\alpha_t) \geq b'$. 

        We will use the Bregman divergence as a potential function to show convergence, so for shorthand, let $g(\alpha) := D_{f}(\alpha_{*} \mid \alpha) = (f(\alpha) - b) + f'(\alpha) (\alpha_{*} - \alpha)$. For any any $t \geq 1$,
        \begin{align*} 
        \alpha_{*} & > \alpha_{t} = \alpha_{t-1} + \frac{b - f(\alpha_{t-1})}{f'(\alpha_{t-1})} \iff f'(\alpha_{t-1}) (\alpha_{*} - \alpha_{t-1}) > b - f(\alpha_{t-1})
        \\ & \iff g(\alpha_{t-1}) = f'(\alpha_{t-1})(\alpha_{*} - \alpha_{t-1}) + f(\alpha_{t-1}) - b'> b - b' .    
        \end{align*}
        Therefore, the first time $g(\alpha_{t-1}) \leq b - b'$, then we will be done in the next iteration. 

        In the sequel, we show that for every step that we have not yet converged, i.e. $f(\alpha_{t}) < b'$, we have 
        \begin{equation} \label{eq:BregmanDrop}
            g(\alpha_t) \leq (1-\eps) g(\alpha_{t-1}).
        \end{equation}
        This combined with the argument above gives the iteration bound. 

        We first recall a classical inequality due to from Radzik \cite{Radzik}: for $t \geq 1$,
        \[ f(\alpha_{t-1}) < b \implies  \qquad \frac{b-f(\alpha_{t})}{b-f(\alpha_{t-1})} + \frac{f'(\alpha_{t})}{f'(\alpha_{t-1})} \leq 1 .  \]      
        This inequality is derived using concavity as follows:
        \[ f(\alpha_{t-1}) - b \leq f(\alpha_{t}) - b + f'(\alpha_{t})( \alpha_{t-1} - \alpha_{t}) = f(\alpha_{t}) - b + f'(\alpha_{t}) \frac{f(\alpha_{t-1})-b}{f'(\alpha_{t-1})} \]
        \[ \implies 1 \geq \frac{f(\alpha_{t})-b}{f(\alpha_{t-1})-b} + \frac{f'(\alpha_{t})}{f'(\alpha_{t-1})} ,  \]
        where we divided by $f(\alpha_{t-1}) - b < 0$.

        Now assuming that $f(\alpha_t) < b'$, $t \geq 1$, by Radzik's inequality we have that
        \begin{equation}
        \label{eq:grad-dop}
        1-\eps \geq 1 - \frac{b-b'}{b-f(\alpha_0)} \geq 1-\frac{b-f(\alpha_t)}{b-f(\alpha_{t-1})} \geq \frac{f'(\alpha_t)}{f'(\alpha_{t-1})}.
        \end{equation}
        Since $f$ is increasing, we conclude that $0 < f'(\alpha_t) \leq (1-\eps) f'(\alpha_{t-1})$. That is, at every non-terminating iteration, the derivative drops by a $1-\eps$ factor. 

        To prove \cref{eq:BregmanDrop}, define $\alpha'_t := \alpha_{t-1} + \frac{b'-f(\alpha_{t-1})}{f'(\alpha_{t-1})} $ for any $t \geq 1$ with $f(\alpha_{t-1}) < b'$. We have 
        \[ f(\alpha'_{t}) \leq f'(\alpha_{t-1})(\alpha'_t-\alpha_{t-1}) + f(\alpha_{t-1}) = b' , \]
        where the first step was by concavity, and the second was by substituting $\alpha'_{t}$. 
        Therefore we can rewrite
        \[
        g(\alpha_{t-1}) = f'(\alpha_{t-1})(\alpha_{*}-\alpha_{t-1}) + f(\alpha_{t-1}) - f(\alpha_{*}) = f'(\alpha_{t-1})(\alpha_{*}-\alpha'_t) , 
        \]
        where we substituted the above expression for $f(\alpha_{*}) = b'$. 
        If $f(\alpha_t) < b'$ for $t \geq 1$, then we can apply the above to show 
        \[ \frac{g(\alpha_{t})}{g(\alpha_{t-1})} = \frac{f'(\alpha_t)(\alpha_{*}-\alpha'_{t+1})}{f'(\alpha_{t-1})(\alpha_{*}-\alpha'_{t})} \leq (1-\eps) \frac{\alpha_{*}-\alpha'_{t+1}}{\alpha_{*}-\alpha'_{t}} \leq 1-\eps ,   \]
        where in the first and last step we substituted the expressions for $g(\alpha_{t}), g(\alpha_{t-1})$ in the numerator and denominator, in the second step we used \cref{eq:grad-dop}, and the final step was by $\alpha'_t \leq \alpha_t \leq \alpha'_{t+1}$. 

        Let $T$ be as in the statement of the theorem, and assume $f(\alpha_{t}) < b'$ for all $t \leq T-1$. Then we have
        \[ g(\alpha_{T-1}) \leq (1-\eps)^{T-1} g(\alpha_{0}) = b - b' ,    \]
        where the first step was by \cref{eq:BregmanDrop} and the second was by definition of $T$. But this implies that $\alpha_{t} \geq \alpha_{*}$, i.e. $f(\alpha_{t}) \geq b'$ as shown above.

    \end{proof}

    We will apply the method with $f = h_{U,z}^{T}$ and targets $b',b = h(1) + \Omega(\gamma)$.   
    The above analysis bounds the number of iterations of Newton-Dinkelbach, depending on the quality of the initial guess. For this purpose, we have \cref{l:updateRange} which gives explicit bounds on the solution. In the following \cref{s:Guess}, we show how to compute an approximation of $\mu_{s}$ in this our setting to give a good initial guess. The guarantee of the guessing algorithm is as follows: 

     \begin{theorem} \label{t:GuessApproximation}
    Algorithm \ref{a:GuessAlgorithm} runs in strongly polynomial time and outputs $\tilde{\mu}$ satisfying
    \[ \sum_{i : \mu_{i} < 1/2} \mu_{i} \leq \tilde{\mu} \leq (1 + 8 n d^{2}) \sum_{i : \mu_{i} < 1/2} \mu_{i} .    \]
    \end{theorem}

    At this point we have all the necessary pieces to prove the guarantees of the update algorithm.

\begin{proof} [Proof of \cref{t:updateGuarantee}]

    Our plan is to apply the Newton-Dinkelbach method for $f := h$ with $b' = h(1) + \gamma/5$, $b = h(1) + \gamma$, and appropriately chosen guess $\alpha_{0} \geq 1$. Note that we are in the feasible case, so by \cref{p:outlineInfeasible}(4), there is a solution $\alpha_{*}$ satisfying $h(\alpha_{*}) = b' = h(1) + \gamma/5$. Therefore, by \cref{t:outlineNDAnalysis}, the number of iterations is bounded by 
    \begin{equation} \label{eq:updateIterBound}
        T := 1 + \left\lceil \log_{\frac{1}{1-\eps}} \left(\max \left\{1, \frac{D_{f}(\alpha_{*} |\alpha_{0})}{b-b'}\right\} \right) \right\rceil , 
    \end{equation} 
    where $\eps$ depends on our choice of $\alpha_{0}$. 

    For the case $h'(1) \geq \gamma/4$, we choose $\alpha_{0} = 1$ and claim that a single iteration suffices. Indeed, we have $\frac{b - b'}{b - h(\alpha_{0})} = \frac{4\gamma/5}{\gamma} = \frac{4}{5} =: \eps$. Further, 
    \[ D_{h}(\alpha_{*} | \alpha_{0}) \leq \alpha_{*} (h(\alpha_{*}) - h(1)) \leq 4 (\gamma/5) ,    \]
    where the first step was by \cref{f:polarizedUpdate}, and in the second step we used the bound $\alpha_{*} \leq 4$ given in \cref{l:updateRange}(1) and the definition $h(\alpha_{*}) = h(1) + \gamma/5$. This gives iteration bound
    \[ T = 1 + \left\lceil \log_{5} \left(\max \left\{1, \frac{4\gamma/5}{4 \gamma/5}\right\} \right) \right\rceil = 1 , 
    \]    
    where we substituted $1/(1-\eps) = 5, D_{h}(\alpha_{*} | \alpha_{0}) \leq 4 \gamma/5$ and $b - b' = 4 \gamma/5$ into \cref{eq:updateIterBound}. 

    In the other case, we use the value $\tilde{\mu}$ approximating $\mu_{s}$ given by \cref{t:GuessApproximation}, and choose guess $\tilde{\alpha} := 1 + \frac{\gamma}{2 \tilde{\mu}}$. 
    In order to apply the Newton-Dinkelbach method, we want to show that this is a feasible guess, i.e. $h(\tilde{\alpha}) \leq b = h(1) + \gamma$. For this, we note $\tilde{\mu} \geq \mu_{s}$ by the lower bound in \cref{t:GuessApproximation}, so we have $\tilde{\alpha} \leq 1 + \frac{\gamma}{2\mu_{s}}$, which implies 
    \[ h(\tilde{\alpha}) \leq h(1) + \gamma ,    \]
    where the first step was by monotonicity of $h$ according to \cref{f:HIncreasingConcave}, and the second step was by the upper bound in \cref{l:updateRange}(2) applied with $\delta = \gamma/2$. 

    Now we argue that $T \lesssim \log(nd)$. 
    As $\tilde{\alpha} \geq 1$, we have $h(\tilde{\alpha}) \geq h(1)$ so we can maintain $\eps := \frac{4}{5} \leq \frac{\gamma - \gamma/5}{\gamma}$. 

    We can bound the Bregman divergence as
    \[ D_{h}(\alpha_{*} | \tilde{\alpha}) \leq \bigg( \frac{\alpha_{*}}{\tilde{\alpha}} - 1 \bigg) (h(\alpha_{*}) - h(\tilde{\alpha})) \leq \frac{1+\gamma/\mu_{s}}{1+\gamma/2\tilde{\mu}}(\gamma/5) \leq O(n d^{2} \gamma ) ,    \]
    where the first step was by \cref{f:polarizedUpdate}, in the second step we used the bound $\alpha_{*} \leq 1+\gamma/\mu_{s}$ from \cref{l:updateRange}(2), and in the final step we used $\tilde{\mu} \leq O(n d^{2}) \mu_{s}$ by \cref{t:GuessApproximation}. This gives iteration bound
    \[ T = 1 + \left\lceil \log_{5} \left(\max \left\{1, \frac{O(n d^{2} \gamma)}{4 \gamma/5}\right\} \right) \right\rceil \leq O(\log (nd)) , 
    \]    
    where we substituted $1/(1-\eps) = 5, D_{h}(\alpha_{*} | \alpha_{0}) \leq O(nd^{2} \gamma)$ and $b - b' = 4 \gamma/5$ into \cref{eq:updateIterBound}. 

    By \cref{t:outlineNDAnalysis}, the update algorithm \ref{a:UpdateAlg} converges in $O(\log(nd))$ iterations of the Newton-Dinkelbach method. To show that this is strongly polynomial, we note that each iteration requires us to evaluate $h,h'$, for which we have the expressions 
    \begin{align*}
        h(\alpha) = \tr[ \alpha M_{T} (\alpha M_{T} + M_{\overline{T}})^{-1}], \quad h'(\alpha) = \tr[ M_{T} (\alpha M_{T} + M_{\overline{T}})^{-1} M_{\overline{T}} (\alpha M_{T} + M_{\overline{T}})^{-1} ] , 
    \end{align*}
    with $M_{T} = U_{T} Z_{T} U_{T}^{\Tr}, M_{\overline{T}} = U_{\overline{T}} Z_{\overline{T}} U_{\overline{T}}^{\Tr}$ as discussed in \cref{r:hStronglyPoly}. 
    \end{proof}

    The above analysis shows that in each iteration the update is strongly polynomial. 
    In order for our full algorithm to be strongly polynomial, we will the following useful bounds on the updates in each iteration. 

    \begin{corollary} \label{c:updateBound}
        For input $(U,z,T,\gamma)$ as in \cref{t:updateGuarantee}, the update $\hat{\alpha}$ satisfies
        \[ 1 \leq \hat{\alpha} \leq 1 + O(\gamma/\mu_{\min}) ,    \]
        where $\mu := {\rm spec}(U_{T} Z_{T} U_{T}^{\Tr} (U Z U^{\Tr})^{-1})$ and $\mu_{\min}$ denotes the smallest non-zero eigenvalue. 
    \end{corollary}
    \begin{proof}
    In the case when $h'(1) \geq \gamma/4$, the proof of \cref{t:updateGuarantee} shows that the Newton-Dinkelbach method only requires a single iteration, so we can bound
    \[ \hat{\alpha} = \alpha_{0} + \frac{\gamma}{h'(\alpha_{0})} \leq 1 + \frac{\gamma}{\gamma/4} = 5 ,    \]
    where the first step was by definition of the Newton step, and in the second step we substituted $\alpha_{0} = 1$ and $h'(1) \geq \gamma/4$ by assumption. 

    In the other case, $\hat{\alpha} = \alpha_{t}$ for $t$ the last iteration of the Newton-Dinkelbach method. Note $h(\alpha_{t-1}) \leq h(1) + \gamma/5 = h(\alpha_{*})$ as otherwise we would be done in the previous iteration, so we have $\alpha_{t-1} \leq \alpha_{*} \leq 1 + \gamma/\mu_{s}$ by \cref{l:updateRange}(2), where $\mu_{s} := \sum_{i : \mu_{i} < 1/2} \mu_{i}$. To bound the denominator, we have
    \[ h'(\alpha_{t-1}) = \sum_{i=1}^{d} \frac{\mu_{i} (1-\mu_{i})}{(1 + (\alpha_{t-1}-1) \mu_{i})^{2}} 
    \geq \sum_{i=1}^{d} \frac{\mu_{i}(1-\mu_{i})}{(1 + \gamma \mu_{i}/\mu_{s})^{2}} \geq \sum_{i : \mu_{i} < 1/2} \frac{\mu_{i}/2}{(1+\gamma)^{2}} \geq \frac{\mu_{s}}{8} ,  \]
    where the first step was by the expression in \cref{l:outlineHParts}, in the second step we used $\alpha_{t-1} \leq 1 + \gamma/\mu_{s}$ as shown above, in the third step we consider the $\mu_{i} < 1/2$ terms to bound $1-\mu_{i} \geq 1/2$ in the numerator and $\mu_{i} \leq \mu_{s}$ in the denominator, and the final step was by $\gamma \leq 1$. Putting this together, we bound the output 
    \[ \hat{\alpha} = \alpha_{t-1} + \frac{h(1) + \gamma - h(\alpha_{t-1})}{h'(\alpha_{t-1})} \leq 1 + \frac{\gamma}{\mu_{s}} + \frac{\gamma}{h'(\alpha_{t-1})} \leq 1 + \frac{\gamma}{\mu_{s}} + \frac{8 \gamma}{\mu_{s}} ,   \]
    where the first step was by definition of the Newton-Dinkelbach method, in the second step we used that $h(\alpha_{t-1}) \geq h(1)$ and $\alpha_{t-1} \leq \alpha_{*} \leq 1 + \gamma/\mu_{s}$ as shown above, and the final step was by the bound $h'(\alpha_{t-1}) \geq \mu_{s}/8$ shown above. The theorem follows as $\mu_{s} := \sum_{i : \mu_{i} < 1/2} \geq \mu_{\min} > 0$. 
    \end{proof}

In the following section we prove the approximation guarantee required for the starting guess.

 \section{Guessing Starting Point} \label{s:Guess}

    In this section, we present the guessing algorithm in \ref{a:UpdateAlg} and prove its approximation guarantee:

\begin{algorithm} \label{a:GuessAlgorithm}
\caption{$\operatorname{Approx-Small-Eigen-Sum}(U,z,T)$}

\KwIn{$U \in \R^{d \times n}, \rk(U)=d, z \in \R^{n}_{++}, T \subseteq [n]$ with $\mu := {\rm spec}( U_{T} Z_{T} U_{T}^{\Tr} (U Z U^{\Tr})^{-1})$, $\sum_{i=1}^d \mu_i(1-\mu_i) < 1/4$;}
\KwOut{$\tilde{\mu}$ satisfying $\sum_{i: \mu_i < 1/2} \mu_{i} \leq \tilde{\mu} \leq (1+ 8 n d^2) \sum_{i: \mu_i < 1/2} \mu_i$}
$p \gets \lfloor \tr[ U_{T} Z_{T} U_{T}^{\Tr} (U Z U^{\Tr})^{-1} ] \rceil$\;
\If{$p=\rk(V_{T})$}{Output $0$\;}
\If{$p=0$}{Output $\tr[ U_{T} Z_{T} U_{T}^{\Tr} (U Z U^{\Tr})^{-1} ]$\;}
\Else{
$D \gets \text{DetLocalOpt}(U,z,T,p)$ (Algorithm \ref{a:DetLocalOpt})\;
Output $\tr[(I_{d} - U_{D} (U_{D}^{\Tr} (U Z U^{\Tr})^{-1} U_{D})^{-1} U_{D}^{\Tr} (U Z U^{\Tr})^{-1} ) U_{T} Z_{T} U_{T}^{\Tr} (U Z U^{\Tr})^{-1} ]$\; }

\end{algorithm}

     \begin{theorem} [Restatement of \cref{t:GuessApproximation}] \label{t:GuessApproxR}
    Algorithm \ref{a:GuessAlgorithm} outputs value $\tilde{\mu}$ satisfying 
    \[ 
       \sum_{i: \mu_i < 1/2} \mu_{i} \leq \tilde{\mu} \leq (1+ 8 n d^{2}) \sum_{i: \mu_i < 1/2} \mu_{i} , 
    \]
    where $\mu := {\rm spec}(U_{T} Z_{T} U_{T}^{\Tr} (U Z U^{\Tr})^{-1})$. Furthermore, the algorithm runs in strongly polynomial time and uses $O(n d^{2} \log n)$ matrix computations involving $d \times d$ matrices (multiplication, determinant, inverse). 
\end{theorem}

    Throughout this section we will use shorthand $V := (U Z U^{\Tr})^{-1/2} U Z^{1/2}$. Note that $V V^{\Tr} = I_{d}$, $V^{\Tr} V$ is the orthogonal projection onto ${\rm im}(V^{\Tr}) = {\rm im}(\sqrt{Z} U^{\Tr})$, and 
    \[ {\rm spec}(V_{T} V_{T}^{\Tr}) = {\rm spec}((U Z U^{\Tr})^{-1/2} U_{T} Z_{T} U_{T}^{\Tr} (U Z U^{\Tr})^{-1/2} ) = {\rm spec}( U_{T} Z_{T} U_{T}^{\Tr} (U Z U^{\Tr})^{-1}) = \mu . \]

    Ignoring the edge cases $p \in \{0,\rk(U_{T})\}$, there are essentially two main steps in the algorithm, whose guarantees we prove in the following two subsections: (1) finding a local optimum for subdeterminant, and (2) computing the norm of the vectors after projection. In the remainder of the section, we motivate this algorithm and prove the approximation result assuming the guarantees of these two steps. 

    Our goal is to compute the sum of eigenvalues $\mu_{s} := \sum_{i : \mu_{i} < 1/2} \mu_{i}$ for $V_{T} V_{T}^{\Tr}$ in strongly polynomial time, where $\mu = {\rm spec}( V_{T}^{\Tr} V_{T})$ as defined in \cref{l:outlineHParts}. 
    The following characterization of spectral projections gives some intuition for our guessing algorithm. 

    \begin{theorem} [Ky-Fan] \label{t:KyFan}
    For any symmetric matrix with eigendecomposition $X = \sum_{i=1}^{d} \lambda_{i} e_{i} e_{i}^{\Tr}$ and $\lambda \in \R^{d}$ in non-increasing order,  
    \[  \sum_{i=1}^{k} \lambda_{i}(X) = \sum_{i=1}^{k} \langle e_{i}, X e_{i} \rangle = \max_{\rk(P) = k} \tr[P X] ,   \]
    where the max runs over all orthogonal projections of rank $k$. 
    \end{theorem}

    Letting $X := V_{T} V_{T}^{\Tr}$, this tells us that we can approximate the spectral sum $\mu_{s}$ if we can guess the number $p := |\{i : \mu_{i} \geq 1/2\}|$ and then compute a projector which approximates the projector onto the top $p$ eigenvalues.     
    The conditions $\mu \in [0,1]^{d}$ and $\sum_{i=1}^{d} \mu_{i} (1- \mu_{i}) < 1/4$ guarantee that all of the eigenvalues are either close to $0$ or close to $1$, so $\tr[V_{T} V_{T}^{\Tr}] \approx p$.
    Furthermore, in the gapped setting, all vectors have very little mass away from the large eigenspace. Therefore we could hope that the span of a ``representative" column subset gives a good approximation of the large eigenspace. 

    Formally, we want to choose a subset $D$ such that all vectors $v_{j \in T}$ are close to ${\rm im}(V_{D})$ in Euclidean distance, and we can use the determinant as a proxy for this property. Therefore, we will compute an approximate local optimum for subdeterminants, and use this as our ``representative" set. 

    \begin{definition} \label{d:detLocalOpt}
        For input $\Pi \in \R^{m \times m}$, $D \subseteq [m]$ is a $\beta$-approximate local optimum for subdeterminant if 
        \[ \forall i \in D, j \not\in D: \qquad \det( \Pi_{D-i+j, D-i+j} ) \leq \beta \det(\Pi_{D,D}) , \]
        where $\Pi_{S,S}$ denotes the $S \times S$ principal submatrix. 
    \end{definition}
    
    In \cref{ss:localOpt}, we show how to compute a $\beta$-local optimum for the subdeterminant of $U_{T}^{\Tr} (U Z U^{\Tr})^{-1} U_{T} Z_{T}$, or equivalently $V_{T}^{\Tr} V_{T}$, using a classical algorithm of Knuth \cite{Knuth}. Then, in \cref{ss:ProjAnalysis} we show that this local optimum subdeterminant property suffices to approximate $\mu_{s}$ in our setting.

    We can now give the proof of our approximation result assuming these two components. 

    \begin{proof} [Proof of \cref{t:GuessApproxR}]
    Recall $V := (U Z U^{\Tr})^{-1/2} U Z^{1/2}$ for shorthand with $V V^{\Tr} = I_{d}$ and $\mu = {\rm spec}(V_{T} V_{T}^{\Tr})$. 
    We compute $p := \lfloor \tr[V_{T} V_{T}^{\Tr}] \rceil$ in the first step of our algorithm, and we want to show that $p = |\{i : \mu_{i} \geq 1/2\}|$. We have 
    \[ \frac{1}{4} > \sum_{i=1}^{d} \mu_{i} (1-\mu_{i}) \geq \sum_{i : \mu_{i} < 1/2} \frac{\mu_{i}}{2} + \sum_{i : \mu_{i} \geq 1/2} \frac{(1-\mu_{i})}{2} ,     \]
    where the first step was by assumption, and in the second step we used $1-\mu_{i} \geq 1/2$ for the first sum and $\mu_{i} \geq 1/2$ for the second sum. Therefore, we can bound
    \begin{align*}
        \tr[V_{T} V_{T}^{\Tr}] & = \sum_{i=1}^{d} \mu_{i} \leq \sum_{i : \mu_{i} \geq 1/2} 1 + \sum_{i : \mu_{i} < 1/2} \mu_{i} < |\{i: \mu_{i} \geq 1/2\}| + 1/2
        \\ \tr[V_{T} V_{T}^{\Tr}] & = \sum_{i=1}^{d} \mu_{i} \geq \sum_{i : \mu_{i} \geq 1/2} (1 - (1 - \mu_{i})) + \sum_{i : \mu_{i} < 1/2} 0 > |\{i: \mu_{i} \geq 1/2\}| - 1/2 , 
    \end{align*}
    where in both lines the first step was by definition of $\mu$, in the second step we used $0 \preceq V_{T} V_{T}^{\Tr} \preceq I_{d}$ so $\mu_{i} \in [0,1]$, and the final step was by the calculation above. This shows that the nearest integer to $\tr[V_{T} V_{T}^{\Tr}]$ is exactly the number of large eigenvalues $|\{i: \mu_{i} \geq 1/2\}|$. 
    
    Now assuming we have the correct value of $p$, we first consider the easy cases: if $p=0$ then there are no large eigenvalues and we output $\tilde{\mu} = \tr[V_{T}V_{T}^{\Tr}] = \mu_{s}$; similarly if $p = \rk(V_{T})$ then there are no small eigenvalues and we output $\tilde{\mu} = 0 = \mu_{s}$.

    In the remaining case, we can apply the analysis of \cref{t:LocalOptAnalysis} and \cref{p:projGuarantee} to our input $V_{T}$ with $\tr[V_{T} V_{T}^{\Tr}] \geq p - 1/2$ to give the approximation guarantee in the theorem. 
    \end{proof}

    In the remainder of the section, we show how to compute an approximate local optimum for the determinant, and give guarantees the projection step. We will use the following linear algebraic facts in our analysis. 

    \begin{fact} \label{f:ProjVariational}
        For $V \in \R^{d \times k}$ and vector $u \in \R^{d}$, 
        \[ \|(I_{d} - P_{V}) u\|_{2} = \min_{y \in \R^{k}} \|V y + u\|_{2} ,    \]
        where $P_{V} := V (V^{\Tr} V)^{-1} V^{\Tr}$ is the orthogonal projection onto ${\rm im}(V)$. 
    \end{fact}

    \begin{fact} \label{f:DetGS}
        For linearly independent $V \in \R^{d \times k}$ and vector $v \in \R^{d}$, let $\tilde{V} := [V, v]$ be the concatenation. Then
        \[ \det(\tilde{V}^{\Tr} \tilde{V}) = \det(V^{\Tr} V) \|(I_{d} - P_{V}) v\|_{2}^{2} ,   \]
        where $P_{V} := V (V^{\Tr} V)^{-1} V^{\Tr}$ is the orthogonal projection onto ${\rm im}(V)$. 
    \end{fact}

    \begin{lemma} \label{l:detAvgLB}
        For matrix $V \in \R^{d \times m}$ with column subset $V_{S}$ of full column rank,  
        \[ \max_{j \in [m]} \frac{\det(V_{S+j}^{\Tr} V_{S+j})}{\det(V_{S}^{\Tr} V_{S})} = \max_{j \in [m]} \|(I_{d} - P_{S}) v_{j}\|_{2}^{2} \geq \frac{1}{m - |S|} \tr[(I_{d} - P_{S}) V V^{\Tr}] ,    \]
        where $P_{S} := V_{S} (V_{S}^{\Tr} V_{S})^{-1} V_{S}$ is the orthogonal projection onto ${\rm im}(V_{S})$. 
    \end{lemma}
    \begin{proof}
        The first equality $\det(V_{S+j}^{\Tr} V_{S+j}) = \det(V_{S}^{\Tr} V_{S}) \|(I_{d} - P_{S}) v_{j}\|_{2}^{2}$ is due to \cref{f:DetGS}. For the inequality, we can bound
        \[ \max_{j \in [m]} \|(I_{d} - P_{S}) v_{j}\|_{2}^{2} \geq \frac{1}{m - |S|} \sum_{j \not\in S} \|(I_{d} - P_{S}) v_{j}\|_{2}^{2} = \frac{1}{m - |S|} \tr[(I_{d} - P_{S}) V V^{\Tr}] ,   \]
        where the first step was by averaging, and in the second step we used that $I_{d} - P_{S}$ is a projection and that $(I_{d} - P_{S}) v_{j} = 0$ if $j \in S$. 
    \end{proof}

 \subsection{Local Optimum} \label{ss:localOpt}

    In this section we present the algorithm to compute the required local optimum satisfying \cref{d:detLocalOpt} and show it runs in strongly polynomial time.

 \begin{algorithm}[h] \label{a:DetLocalOpt}
 \caption{$\operatorname{DetLocalOpt}(U,z,T,p)$}

\KwIn{$U \in \R^{d \times n}, \rk(U)=d, z \in \R^{n}_{++}, T \subseteq [n]$, $0 < p < \rk(U_{T})$, $\tr[U_{T} Z_{T} U_{T}^{\Tr} (U Z U^{\Tr})^{-1}] \geq p - 1/2$ }
\KwOut{ $\beta=2$-local maximizer for determinant $D \in {T \choose p}$ according to \cref{d:detLocalOpt} }
\For{$k = 1, ... p$}
{ $G \gets \arg\max_{i \in T-G} \det(U_{G+i}^{\Tr} (U Z U^{\Tr})^{-1} U_{G+i} Z_{G+i})$\; }
\While{True}{ 
$(i,j) \gets \arg\max_{i \in D, j \in T-D} \det(U_{D-i+j}^{\Tr} (U Z U^{\Tr})^{-1} U_{D-i+j} Z_{D-i+j})$\;
\If{$\det(U_{D-i+j}^{\Tr} (U Z U^{\Tr})^{-1} U_{D-i+j} Z_{D-i+j}) \leq 2 \det(U_{D}^{\Tr} (U Z U^{\Tr})^{+} U_{D} Z_{D})$}{Output $D$\;}
$D \gets D \triangle \{i,j\}$\;
}
\end{algorithm}

    The first part is a simply greedy maximization, and the second is a local search. This strongly polynomial algorithm appears in \cite{Kachiyan}, where it is used (along with a strongly polynomial algorithm for approximate John's position) to compute an exponential approximation for max subdeterminant. The guarantees are as follows:

    \begin{prop} \label{t:LocalOptAnalysis}
        The Greedy and Local Search Algorithm \ref{a:DetLocalOpt} applied to input $(U,z,T,p)$ computes a $\beta=2$-approximate local maximum according to \cref{d:detLocalOpt}. 
        This can be implemented using $O(n d + n d^{2} \log n)$ determinant computations on sub-matrices of $U_{T}^{\Tr} (U Z U^{\Tr})^{-1} U_{T} Z_{T}$. 
    \end{prop}
    \begin{proof}
        Recall that we use $V := (U Z U^{\Tr})^{-1/2} U Z^{1/2}$ for shorthand, and note that for any $S \subseteq [n]$
        \[ \det(V_{S}^{\Tr} V_{S}) = \det(U_{S}^{\Tr} (U Z U^{\Tr})^{-1} U_{S} Z_{S})   \]
        by multiplicativity of determinant. Therefore we can equivalently analyze our algorithm in terms of $V$. 
        
        In each iteration that the local search algorithm finds an improvement, the subdeterminant increases by a factor of $\beta=2$. As $V_{T} V_{T}^{\Tr} \preceq V V^{\Tr} = I_{d}$, we must have that all subdeterminants are upper bounded by $1$. Therefore the runtime bound holds if we can give $\exp(-\poly(d,n))$ lower bound on the initial subdeterminant. 
        For the $k$-th greedy iteration, we can lower bound
        \begin{align*} 
        \frac{\det(V_{G_{k+1}}^{\Tr} V_{G_{k+1}})}{\det(V_{G_{k}}^{\Tr} V_{G_{k}})} 
        & = \max_{j \in T-G_{k}} \frac{\det(V_{G_{k}+j}^{\Tr} V_{G_{k}+j})}{\det(V_{G_{k}}^{\Tr} V_{G_{k}})} \geq \frac{1}{|T| - k} \tr[(I_{d} - P_{G_{k}}) V_{T} V_{T}^{\Tr}]
        \\ & \geq \frac{1}{|T|-k} \big( \tr[V_{T} V_{T}^{\Tr}] - \|P_{G_{k}}\|_{1} \|V_{T} V_{T}^{\Tr}\|_{\ope} \big) \geq \frac{p - k - 1/2}{|T|-k} , 
        \end{align*}
        where the first step was by definition of the greedy algorithm, in the second step we applied \cref{l:detAvgLB}, in the third step we applied Holder's inequality for the second term, and in the final step we used the conditions $\tr[V_{T} V_{T}^{\Tr}] \geq p - 1/2$, $V_{T} V_{T}^{\Tr} \preceq I_{d}$, and $\|P\|_{1} \leq \rk(P) = k$ for orthogonal projection $P := P_{G_{k}}$. 

        Taking the product of the above bound for $k = 0, ..., p-1$ shows that the greedy initialization satisfies 
        \[ \det(V_{G}^{\Tr} V_{G}) \geq \frac{1}{2p} {|T| \choose p}^{-1}     \]
        As described above, $V_{T} V_{T}^{\Tr} \preceq I_{d}$ so the maximum subdeterminant is $\leq 1$, which implies the local search algorithm starting with $V_{G}$ finds $\beta = 2$-local maximizer in a bounded number of iterations
    \[ \log_{\beta} \frac{1}{\det(V_{G}^{\Tr} V_{G})} \leq \log_{2} 2p {|T| \choose p} \leq O(p \log n)  . \]

    For the runtime, the greedy algorithm requires $p$ iterations, each with $\leq |T|$ determinant computations; the local search algorithm requires $O(p \log |T|)$ iterations, each with $p (|T| - p)$ determinant computations. 
    \end{proof}

\subsection{Analysis of Projection Step} \label{ss:ProjAnalysis}

In this subsection we analyze the output of the projection step in Algorithm \ref{a:GuessAlgorithm} assuming that we have an approximate local optimum to the subdeterminant according to \cref{d:detLocalOpt}. 

\begin{prop} \label{p:projGuarantee}
    Consider input $(U,z,T,p)$ for Algorithm \ref{a:GuessAlgorithm} and let $D \subseteq T$ be the $\beta=2$-local optimum according to \cref{d:detLocalOpt}. The output of the algorithm satisfies 
    \[  \sum_{i = p+1}^{d} \mu_{i} \leq  \tr[(I_{d} - U_{D} ( U_{D}^{\Tr} ( U Z U^{\Tr})^{-1} U_{D})^{-1} U_{D}^{\Tr} (U Z U^{\Tr})^{-1} ) U_{T} Z_{T} U_{T}^{\Tr} (U Z U^{\Tr})^{-1} ] \leq (1 + 8nd^{2}) \sum_{i=p+1}^{d} \mu_{i}   ,  \]
    where $\mu,p$ are defined in the algorithm. 
\end{prop}

    Letting $V := (U Z U^{\Tr})^{-1/2} U Z^{1/2}$ for shorthand, we have $\mu = {\rm spec}(U_{T} Z_{T} U_{T}^{\Tr} (U Z U^{\Tr})^{-1}) = {\rm spec}(V_{T} V_{T}^{\Tr})$, and we can rewrite the above guarantee
    \[ \sum_{i=p+1}^{d} \mu_{i} \leq \tr[ (I_{d} - P_{D}) V_{T} V_{T}^{\Tr}] \leq (1 + 8 nd^{2}) \sum_{i=p+1}^{d} \mu_{i} ,  \]
    where $P_{D} := V_{D} (V_{D}^{\Tr} V_{D})^{-1} V_{D}^{\Tr}$ is the orthogonal projection onto ${\rm im}(V_{D})$. So intuitively, we want to show that for local optimum $D$, the subspace ${\rm im}(V_{D}) = {\rm im}(P_{D})$ is a good approximation for the top-$p$ eigenspace of $V_{T} V_{T}^{\Tr}$.
    In the sequel, we use eigendecomposition $V_{T} V_{T}^{\Tr} = \sum_{i=1}^{d} \mu_{i} e_{i} e_{i}^{\Tr}$, and $E_{b} := [e_{1}, ..., e_{p}] \in \R^{d \times p}$ and $E_{s} := [e_{p+1}, ..., e_{d}] \in \R^{d \times (d-p)}$ the matrices corresponding to big and small eigenvectors, respectively. 

We first show the vectors $V_{D}$ are well-invertible for local optimum $D$. 

\begin{lemma} \label{l:sigmaLB}
    Let $(V_{T},p)$ be as defined above and $D \in {T \choose p}$ be a $\beta=2$-local maximum for the subdeterminant of $V_{T}^{\Tr} V_{T}$ according to \cref{d:detLocalOpt}. Then, 
    \[ \sigma_{p}(V_{D})^{2} \geq \frac{1}{4 np} , \qquad \text{ and } \qquad \sigma_{p}(E_{b}^{\Tr} V_{D})^{2} \geq \sigma_{p}(V_{D})^{2} - \|E_{s}^{\Tr} V_{D}\|_{\ope}^{2} ,  \]
    for $E_{b}, E_{s}$ the eigenvector matrices of $V_{T} V_{T}^{\Tr}$ defined above
\end{lemma}
\begin{proof}
    The second inequality follows simply as $[E_{b}, E_{s}]$ is an orthonormal basis for $\R^{d}$:
    \[ \sigma_{p}(E_{b}^{\Tr} V_{D})^{2} = \min_{x \in \R^{D}} \frac{\|E_{b}^{\Tr} V_{D} x\|_{2}^{2}}{\|x\|_{2}^{2}} = \min_{x \in \R^{D}} \frac{\|V_{D} x\|_{2}^{2} - \|E_{s}^{\Tr} V_{D} x\|_{2}^{2}}{\|x\|_{2}^{2}} \geq \sigma_{p}(V_{D})^{2} - \|E_{s}^{\Tr} V_{D}\|_{\ope}^{2},     \]
    where we used $|D| = p$ so $\|V_{D} x\|_{2} \geq \sigma_{p}(V_{D}) \|x\|_{2}$ and $\|E_{s}^{\Tr} V_{D} x\|_{2} \leq \|E_{s}^{\Tr} V_{D}\|_{\ope} \|x\|_{2}$ by definition. 

    In the rest, we prove the first inequality. 
    We first show that the columns of $V_{D}$ are ``well-separated" in the sense that each is far from the hyperplane spanned by the rest. For any $S \subseteq T$ such that $V_{S}$ has linearly independent columns, we let $P_{S} := V_{S} (V_{S}^{\Tr} V_{S})^{-1} V_{S}^{\Tr}$ be the orthogonal projection onto ${\rm im}(V_{S})$ and $P_{S}^{\perp} := I_{d} - P_{S}$ be the orthogonal complement. Then by local optimality, for any $i \in D, j \not\in D$ we have
    \[ 2 \det(V_{D-i}^{\Tr} V_{D-i}) \|P_{D-i}^{\perp} v_{i}\|_{2}^{2} = 2 \det(V_{D}^{\Tr} V_{D})
    \geq \max_{j \in T} \det(V_{D-i+j}^{\Tr} V_{D-i+j}) = \max_{j \in T} \det(V_{D-i}^{\Tr} V_{D-i}) \|P_{D-i}^{\perp} v_{j}\|_{2}^{2} ,    \]   
    where the first and last step was by \cref{f:DetGS}, and the middle inequality was by $\beta=2$ local optimality of $D$ according to \cref{d:detLocalOpt}. This then implies
    \begin{align*} 
    \|P_{D-i}^{\perp} v_{i}\|_{2}^{2} \geq \frac{\max_{j \in T} \|P_{D-i}^{\perp} v_{j}\|_{2}^{2}}{2} \geq \frac{\tr[(I_{d} - P_{D-i}) V_{T} V_{T}^{\Tr}]}{2(|T| - (p-1))} 
    \geq \frac{(p-1/2) - (p-1)}{2n} \geq \frac{1}{4n} ,    
    \end{align*}
    where the first step was by the calculation above, in the second step we applied the lower bound in \cref{l:detAvgLB}, in the third step we used lower bound $\tr[V_{T} V_{T}^{\Tr}] \geq p - 1/2$ by assumption, and applied upper bound 
    \[ \tr[P_{D-i} V_{T} V_{T}^{\Tr}] \leq \|P_{D-i}\|_{1} \|V_{T} V_{T}^{\Tr}\|_{\ope} \leq p-1 \]
    by Holder's inequality, $\|P_{D-i}\|_{1} = \rk(P_{D-i}) = |D|-1 = p-1$ and $V_{T} V_{T}^{\Tr} \preceq I_{d}$. 
    
    This implies a lower bound on the $p$-th singular value of $V_{D}$ as follows: for $x \in \R^{D}$ with $|x_{i}| = \|x\|_{\infty}$, 
    \[ \| V_{D} x\|_{2}^{2} \geq \|P_{D-i}^{\perp} V_{D} x\|_{2}^{2} = \|P_{D-i}^{\perp} x_{i} v_{i}\|_{2}^{2} \geq \frac{|x_{i}|^{2}}{4 n} \geq \frac{\|x\|_{2}^{2}}{4 n p} ,    \]
    where the inequality in the first step is because $P_{D-i}^{\perp}$ is an orthogonal projection, in the second step we used that $P_{D-i}^{\perp} v_{j} = 0$ for all $j \in D-i$, in the third step we used the lower bound $\|P_{D-i}^{\perp} v_{i}\|_{2}^{2} \geq 1/4 n$ shown above, and the final step was by our assumption $|x_{i}| = \|x\|_{\infty}$ so $\|x\|_{2}^{2} \leq |D| \|x\|_{\infty}^{2} = p |x_{i}|^{2}$. 
    Since $x \in \R^{D}$ was arbitrary, this gives the required lower bound $\sigma_{p}(V_{D})^{2} \geq 1/4np$.
\end{proof}

Next, we bound the projection in terms of the overlap of $V_{D}$ with the large eigenspace. 

\begin{lemma} \label{l:projBoundSigma}
    For $(V,T,p,D)$ as defined above 
    \[ \tr[(I_{d} - P_{D}) V_{T} V_{T}^{\Tr}] \leq (1 + \|(E_{b}^{\Tr} V_{D})^{-1}\|_{\ope}^{2}) \sum_{i=p+1}^{d} \mu_{i} ,    \]
    where $\mu = {\rm spec}(V_{T} V_{T}^{\Tr})$ and $E_{b} \in \R^{d \times p}$ is the matrix of top-$p$ eigenvectors.  
\end{lemma}
\begin{proof}
    If $E_{b}^{\Tr} V_{D} \in \R^{p \times p}$ is not invertible, then the lemma is trivial as the RHS is $\infty$. 
    Otherwise, we show that the quantity $\|E_{b}^{\Tr} V_{D}^{-1}\|_{\ope}$ gives an upper bound on the distance between subspaces ${\rm im}(E_{b})$ and ${\rm im}(V_{D})$. 
    
    Let $P_{D} := V_{D} (V_{D}^{\Tr} V_{D})^{-1} V_{D}^{\Tr}$ be the orthogonal projection onto ${\rm im}(V_{D})$ and $P_{D}^{\perp} := I_{d} - P_{D}$ be the orthogonal complement. Now for arbitrary $a \in {\rm im}(E_{b})$,
\[ a^{\Tr} P_{D}^{\perp} a  = \| P_{D}^{\perp} a\|_{2}^{2} = \min_{x \in \R^{D}} \|V_{D} x - a\|_{2}^{2} \leq \|V_{D} (E_{b}^{\Tr} V_{D})^{-1} (E_{b}^{\Tr} a) - a\|_{2}^{2} = \|E_{s}^{\Tr} V_{D} (E_{b}^{\Tr} V_{D})^{-1} E_{b}^{\Tr} a\|_{2}^{2} ,   \]
where the first step was because $P_{D}^{\perp}$ is an orthogonal projection, in the second step we applied \cref{f:ProjVariational}, in the third step we substituted $x = (E_{b}^{\Tr} V_{D})^{-1} E_{b}^{\Tr} a $ for invertible $E_{b}^{\Tr} V_{D} \in \R^{p \times p}$, and in the final step we used that $E_{b}^{\Tr} V_{D} (E_{b}^{\Tr} V_{D})^{-1} (E_{b}^{\Tr} a) = E_{b}^{\Tr} a$. Therefore, we can bound
\begin{align*}
    \tr[P_{D}^{\perp} V_{T} V_{T}^{\Tr}] & = \sum_{i=1}^{d} \mu_{i} e_{i}^{\Tr} P_{D}^{\perp} e_{i} 
    \leq \sum_{i=1}^{p} \|E_{s}^{\Tr} V_{D} (E_{b}^{\Tr} V_{D})^{-1} E_{b}^{\Tr} e_{i}\|_{2}^{2} + \sum_{i=p+1}^{d} \mu_{i}
    = \|(E_{s}^{\Tr} V_{D}) (E_{b}^{\Tr} V_{D})^{-1}\|_{F}^{2} + \sum_{i > p} \mu_{i} , 
\end{align*}
where the first step was by eigendecomposition $V_{T} V_{T}^{\Tr} = \sum_{i=1}^{d} \mu_{i} e_{i} e_{i}^{\Tr}$, in the second step we used the calculation above for $e_{i \in [p]} \in {\rm im}(E_{b})$ to bound the first $p$ terms and $\langle e_{i}, P_{D}^{\perp} e_{i} \rangle \leq 1$ for the remaining terms, and in the last step we used $\{E_{b}^{\Tr} e_{1}, ..., E_{b}^{\Tr} e_{p}\} = I_{p}$. We can complete the proof as
\[ \|(E_{s}^{\Tr} V_{D}) (E_{b}^{\Tr} V_{D})^{-1}\|_{F}^{2} \leq \|E_{s}^{\Tr} V_{D}\|_{F}^{2} \|(E_{b}^{\Tr} V_{D})^{-1}\|_{\ope}^{2} \leq \|E_{s}^{\Tr} V_{T}\|_{F}^{2} \|(E_{b}^{\Tr} V_{D})^{-1}\|_{\ope}^{2} = \|(E_{b}^{\Tr} V_{D})^{-1}\|_{\ope}^{2} \sum_{i=p+1}^{d} \mu_{i},   \]
where the first step was by sub-multiplicativity $\|X Y\|_{F} \leq \|X\|_{F} \|Y\|_{\ope}$, in the second step we used that $V_{D}$ is a submatrix of $V_{T}$, and the final step was because $E_{s}$ is the matrix of eigenvectors of $V_{T} V_{T}^{\Tr}$ corresponding to small eigenvalues $\{u_{p+1}, ..., \mu_{d}\}$. 

\end{proof}

Combining these two steps gives the approximation guarantee for the projected norm. 

\begin{proof} [Proof of \cref{p:projGuarantee}]
    The lower bound follows from the Ky-Fan Theorem \ref{t:KyFan}     
    \[  \sum_{i=1}^{p} \mu_{i} = \max_{\rk(P) = p} \tr[P V_{T} V_{T}^{\Tr}] \geq \tr[P_{D} V_{T} V_{T}^{\Tr}] ,      \]
    as $P_{D}$ is an orthogonal projection with $\rk(P_{D}) = p$. 

    If $\mu_{p+1} > 1/8nd$, then the upper bound is trivial as
\[ \|P_{D}^{\perp} V_{T}\|_{F}^{2} \leq \|V_{T}\|_{F}^{2} \leq d \leq (8 nd^{2}) \mu_{p+1} \leq 8 n d^{2} \sum_{i=p+1}^{d} \mu_{i} ,     \]
where in the second step we used the assumption $V_{T} V_{T}^{\Tr} \preceq I_{d}$, and in the third we used $\mu_{p+1} > 1/8nd$. 

    In the remaining case $\mu_{p+1} \leq 1/8nd$, we can apply the analysis of \cref{l:projBoundSigma}. 
    For this, we bound 
    \[ \|(E_{b}^{\Tr} V_{D})^{-1}\|_{\ope}^{-2} = \sigma_{p}(E_{b}^{\Tr} V_{D})^{2} \geq \sigma_{p}(V_{D})^{2} - \|E_{s}^{\Tr} V_{D}\|_{\ope}^{2} \geq \frac{1}{4nd} - \frac{1}{8nd} ,     \]
    where the first step is by definition of singular values for $E_{b}^{\Tr} V_{D} \in \R^{p \times p}$, the second step was by the second inequality from \cref{l:sigmaLB}, and in the final step we used the lower bound $\sigma_{p}(V_{D})^{2} \geq \frac{1}{4nd}$ by \cref{l:sigmaLB} and the upper bound $\|E_{s}^{\Tr} V_{D}\|_{\ope}^{2} = \mu_{p+1} \leq \frac{1}{8nd}$ by assumption. 
    
    Therefore, we can bound the projection as 
    \begin{align*}
        \| P_{D}^{\perp} V_{T}\|_{F}^{2} = \tr[P_{D}^{\perp} V_{T} V_{T}^{\Tr}] \leq (1 + \|(E_{b}^{\Tr} V_{D})^{-1}\|_{\ope}^{2}) \sum_{i=p+1}^{d} \mu_{i} \leq (1 + 8nd) \sum_{i=p+1}^{d} \mu_{i} , 
    \end{align*}
    where the second step was by \cref{l:projBoundSigma} applied to local optimum $D$, and in the final step we used $\|(E_{b}^{\Tr} V_{D})^{-1}\|_{\ope}^{2} \leq 8nd$ as shown above. 
\end{proof}

\section{Regularization} \label{s:regularization}

    In this section we explain the regularization step in Algorithm \ref{a:MainAlgorithm}. This allows us to maintain the property that all intermediate scaling have bit complexity bounded by $\poly(n,b,\log(1/\eps))$, where $b$ is the bit complexity of the input vectors $U \in \R^{d \times n}$. In \cref{ss:RegMagnitude}, we show how to control the magnitude or condition number of the scaling iterates. Then in \cref{ss:RegBitComplexity} we show how this implies strongly polynomial bit complexity by a simple rounding procedure. 
    Finally, we discuss the relation to previous work in \cref{ss:RegPreviousWork}.  

    We begin by presenting a condition measure of frames. 

    \begin{definition} \label{d:rhoMeasure}
    For input $U \in \R^{d \times n}$, let
    \[ \bar{\chi}_{T}(U) := \|((U U^{\Tr})^{-1/2} U_{T})^{+}\|_{\ope} , \qquad \bar{\chi}(U) := \max_{T \subseteq [n]} \bar{\chi}_{T}(U) ,   \]
    \[ \rho_{T}(U) := \|((U U^{\Tr})^{-1/2} U_{T})^{+}\|_{\ope}^{2} - 1 , \qquad \rho(U) := \max_{T \subseteq [n]} \rho_{T}(U) .   \]
    \end{definition}

    Note that these condition measures depend only on ${\rm im}(U^{\Tr})$, as $U^{\Tr} (U U^{\Tr})^{-1/2}$ is invariant under transformations $U \to L U$ for invertible $L \in \R^{d \times d}$.
    
    In \cref{ss:updateRange}, we relate the update step for input $(U,z,T)$ to $\mu = {\rm spec}((U_{T} Z_{T} U_{T}^{\Tr}) (U Z U^{\Tr})^{-1})$.
    We can relate this to our condition number as follows. 

\begin{fact} \label{f:spectralRho}
    For $V \in \R^{d \times n}$ and column subset $T \subseteq [n]$, let $\lambda_{\min}$ denote the smallest non-zero eigenvalue. Then
    \[ (\lambda_{\min}( V_{T}^{\Tr} (V V^{\Tr})^{-1} V_{T} ))^{-1} = \|( V_{T}^{\Tr} (V V^{\Tr})^{-1} V_{T} )^{+}\|_{\ope} = 1 + \rho_{T}(V) .  \] 
\end{fact}

    In our algorithm, we are always dealing with frames $V$ that are column scalings of our original frame $U$. In the next lemma, we show how to relate the condition number $\rho$ between column scalings. 

\begin{lemma} \label{l:rhoZBound}
    For frame $U \in \R^{d \times n}$ and scaling $z \in \R^{n}_{++}$, 
    \[ \forall T \subseteq [n] : \quad  \rho_{T}(U \sqrt{Z}) \leq \frac{\max_{j \not\in T} z_{j}}{\min_{j \in T} z_{j}} \rho_{T}(U)  .  \]
\end{lemma}
\begin{proof}
    For shorthand let $V := U \sqrt{Z}$ and $B := V^{\Tr} (V V^{\Tr})^{-1/2}$ be the orthonormal basis for $W := {\rm im}(V^{\Tr})$. Further, let $B_{T \cdot} \in \R^{T \times d}$ denote the restriction to rows in $T$. We first rewrite $\|(V_{T}^{\Tr}(V V^{\Tr}))^{+} \|_{\ope}$ as
    \begin{align*} \|B_{T \cdot}^{+}\|_{\ope}^{2} 
    & = \max_{y \in {\rm im}(B_{T \cdot})} \frac{\|B_{T \cdot}^{+} y\|_{2}^{2}}{\|y\|_{2}^{2}}
    = \max_{y \in {\rm im}(B_{T \cdot})} \min_{u \in \R^{d}: (B u)_{T} = y} \frac{\|u\|_{2}^{2}}{\|y\|_{2}^{2}} 
    = \max_{y \in {\rm im}(B_{T \cdot})} \min_{ (Bu)_{T} = y} \frac{\|B u\|_{2}^{2}}{\|y\|_{2}^{2}} 
    \\ & = \max_{y \in {\rm im}(V_{T}^{\Tr})} \min_{w \in {\rm im}(V^{\Tr}), w|_{T} = y} \frac{\|w\|_{2}^{2}}{\|y\|_{2}^{2}} 
    = \max_{y \in {\rm im}(V_{T}^{\Tr})} \min_{w \in {\rm im}(V^{\Tr}), w|_{T} = y} \bigg( 1 + \frac{\|w_{\overline{T}}\|_{2}^{2}}{\|w_{T}\|_{2}^{2}} \bigg)
     , 
    \end{align*}
    where the first step was by definition of operator norm, the second step was by definition of the pseudo-inverse, in the third step we used that $B$ has orthonormal columns, i.e. $B^{\Tr} B = I_{d}$, the fourth step was by definition of ${\rm im}(V^{\Tr}) = {\rm im}(B)$, and in the final step we used that $w_{T} = y$. 

    Now note that scaling $U \to U \sqrt{Z} =: V$ gives a scaling of subspaces $w \in {\rm im}(U^{\Tr}) \longleftrightarrow \sqrt{z} \circ w \in {\rm im}(V^{\Tr})$. Therefore, we can upper bound
    \begin{align*} 
    \rho_{T}(V) & = \|B_{T \cdot}^{+}\|_{\ope}^{2} - 1 
    = \max_{y \in {\rm im}(V_{T}^{\Tr})} \min_{w \in {\rm im}(V^{\Tr}), w|_{T} = y} \frac{\|w_{\overline{T}}\|_{2}^{2}}{\|w_{T}\|_{2}^{2}}
    = \max_{y \in {\rm im}(U_{T}^{\Tr})} \min_{w \in {\rm im}(U^{\Tr}), w|_{T} = y} \frac{\|(\sqrt{z} \circ w)_{\overline{T}}\|_{2}^{2}}{\|(\sqrt{z} \circ w)_{T}\|_{2}^{2}}
    \\ & \leq \max_{y \in {\rm im}(U_{T}^{\Tr})} \min_{w \in {\rm im}(U^{\Tr}), w|_{T} = y} \frac{\max_{j \not\in T} z_{j}}{ \min_{j \in T} z_{j}} \frac{\|w_{\overline{T}}\|_{2}^{2}}{\|w_{T}\|_{2}^{2}} 
    = \frac{\max_{j \not\in T} z_{j}}{ \min_{j \in T} z_{j}} \rho_{T}(U) , 
    \end{align*}
    where the first step was by \cref{d:rhoMeasure} and the definition $B := V^{\Tr} (V V^{\Tr})^{-1/2}$, the second step was by the expression above for $\|B_{T \cdot}^{+}\|_{\ope}^{2}$, in the third step we used the correspondence between subspaces ${\rm im}(V^{\Tr}) = \sqrt{Z} {\rm im}(U^{\Tr})$, and the final step by the same calculation as above for $\rho_{T}(U)$.  
\end{proof}

    \begin{remark}
    Above we used the ``lifting map" interpretation of \cite{ScalingInvIPM} for condition measures $\bar{\chi}$ and $\rho$. 
    Specifically, they use $\ell^{{\rm im}(U^{\Tr})}(T) := \sqrt{\rho_{T}(U)}$ to denote the ``lifting score" of the coordinates $T$ with respect to subspace ${\rm im}(U^{\Tr})$. 
    We use notation $\rho$ for simplicity, and point the reader to \cite{ScalingInvIPM} for other interpretations of these condition measures, as well as a more detailed bibliography of the relation to linear programming. 
    \end{remark}

    In \cref{ss:RegBitComplexity}, we will show that $\rho(U)$ is a singly exponential function of dimension and bit complexity of $U$. But in fact, the main message of this section is that $\rho$ is a more refined complexity measure than bit complexity. For example if $A \in \R^{n \times m}$ is the node-edge incidence matrix of a (directed) graph, then $A$ has entries of constant bit complexity, but it can be shown that $\rho(A) \leq \poly(n)$, which is significantly better than the naive $2^{\poly(n)}$ bound derived just using bit complexity. As a further benefit, $\rho$ is continuous and therefore is robust to perturbation. 

        Using these definitions, we can prove a bound on the growth in magnitude of scaling $z$ in a single iteration. 

    \begin{prop} \label{p:iterateUB}
        For input $U$ and scaling $z^{(t)}$,  
        line 8 of Algorithm \ref{a:MainAlgorithm} produces update $z^{(t+1)}$ satisfying
        \[ \|\log z^{(t+1)}\|_{\infty} \leq 2 \|\log z^{(t)}\|_{\infty} + \log \rho(U) + O(1) .    \]
    \end{prop}
    \begin{proof}
    By \cref{c:updateBound}, we have update $z^{(t+1)} = z^{(t)} \circ (1_{\overline{T}} + \hat{\alpha} 1_{T})$ for $\hat{\alpha} \leq O(1/\mu_{\min})$, where $\mu_{\min}$ is the smallest non-zero eigenvalue of $U_{T} Z_{T} U_{T}^{\Tr} (U Z U^{\Tr})^{-1}$ for update set $T \subseteq [n]$. 
  
    We can bound this quantity by
        \[ \frac{1}{\mu_{\min}(Z_{T}^{1/2} U_{T}^{\Tr} (U Z U^{\Tr})^{-1} U_{T} Z_{T}^{1/2})}
        = 1 + \rho_{T}(U \sqrt{Z}) \leq 1 + \bigg( \frac{\max_{j \not\in T} z_{j}}{\min_{j \in T} z_{j}} \bigg) \rho_{T}(U) , 
        \]
        where the first step was by applying \cref{f:spectralRho} for $V := U \sqrt{Z}$, and the second was by \cref{l:rhoZBound}. 
        Without loss, we can assume $z_{\min} = 1$, as uniform scaling of $z$ has no effect on frame scaling, so we can bound the multiplicative term $\frac{z_{\max}}{z_{\min}} = \exp(\|\log z\|_{\infty})$ and $\rho_{T}(U) \leq \rho(U)$ by \cref{d:rhoMeasure}. 
                
        This gives a bound on the update
        \[ \|\log z^{(t+1)}\|_{\infty} - \|\log z^{(t)}\|_{\infty} \leq \log \hat{\alpha} \leq \log \bigg(\frac{O(1)}{\mu_{\min}} \bigg) \leq \|\log z^{(t)}\|_{\infty} + \log \rho(U) + O(1),    \]
        where the second step was by the bound $\hat{\alpha} \leq O(1)/\mu_{\min}$ from \cref{c:updateBound}, and in the final step we used the lower bound on $\mu_{\min}^{-1} \leq \exp(\|\log z\|_{\infty}) \rho(U)$ as shown above. 
    \end{proof}

    This result tells us that if the iterate $z^{(t)}$ is bounded, then the update $z^{(t+1)}$ will also be bounded. Of course, applying this for many iterations could cause the magnitude to blow up. Therefore, in the next subsection we show that for the purpose of approximate frame scaling, we can always regularize our scalings so that the magnitude is bounded by a function of the measure $\rho$. Finally, in \cref{ss:RegBitComplexity} we use a simple rounding procedure to maintain bounded bit complexity of scalings throughout the algorithm.  

\subsection{Controlling Magnitude of Scalings} \label{ss:RegMagnitude}

    In this subsection, we show that we can approximate any marginal using a bounded scaling. Similar results are crucially used in the analysis of previous weakly polynomial time algorithms for frame scaling, 
    and we discuss the relation to these results in \cref{ss:RegPreviousWork}. 
    We believe this result is of independent interest.

\begin{theorem} \label{t:Regularization}
    Given frame $U \in \R^{d \times n}$ and scaling $z \in \R_{++}^{n}$, for any $0 < \delta \leq 1$, there is $\hat{z} \in \R_{++}^{n}$ such that
    \[ \|\lev^{U}(z) - \lev^{U}(\hat{z})\|_{1} \leq 2 n d \delta \qquad \text{and} \qquad \log \frac{\hat{z}_{\max}}{\hat{z}_{\min}} \leq n \log(\rho(U)/\delta) .  \]
\end{theorem}
\begin{proof}
    Assume $z$ is in non-increasing order $z_{1} \geq ... \geq z_{n}$, and apply uniform scaling so that $z_{n} = 1$. For each prefix set $k \in [n]$, if
    \[ \frac{z_{k}}{z_{k+1}} \leq \frac{\rho_{[k]}(U)}{\delta}  ,  \]
    then the current scaling $z$ already satisfies the stated bound, so we are done. Otherwise, we show that we can round down the ratio to this value while incurring little error in leverage scores.   

    So assume there is a set $T := [k]$ such that the above bound is violated. Then we scale down to
    \[ \forall j \leq k: \qquad \hat{z}_{j} := z_{j} \bigg( \frac{z_{k}}{z_{k+1}} \bigg)^{-1} \bigg( \frac{\rho_{T}(U)}{\delta} \bigg)    \]
    and $\hat{z}|_{\overline{T}} = z|_{\overline{T}}$. Note that this rounding is of the form $z = \hat{z} \circ (1_{\overline{T}} + \alpha 1_{T})$ for $\alpha := \frac{z_{k}}{z_{k+1}} \frac{\delta}{\rho_{T}(U)} > 1$. Therefore we can analyze the change in leverage scores using progress function $h := h_{U,\hat{z}}^{T}$ as
    \begin{align*}
        \|\lev^{U}(z) - \lev^{U}(\hat{z})\|_{1} & 
        = \sum_{j \in T} (\lev^{U}_{j}(z) - \lev^{U}_{j}(\hat{z})) + \sum_{j \not\in T} (\lev^{U}_{j}(\hat{z}) - \lev^{U}_{j}(z)) 
        = 2(h(\alpha) - h(1))
        \\ & = 2 \sum_{i=1}^{d} \frac{(\alpha-1) \mu_{i} (1-\mu_{i})}{1 + (\alpha-1) \mu_{i}} 
        \leq 2 \sum_{i=1}^{r} (1-\mu_{i}) \leq 2 d (1 - \mu_{\min}) 
        \end{align*}
    where the first step is because we are scaling up the set $T$ so $\lev^{U}_{j}(z) \geq \lev^{U}_{j}(\hat{z})$ for $j \in T$ and $\lev^{U}_{j}(z) \leq \lev^{U}_{j}(\hat{z})$ for $j \not\in T$, the second step was by \cref{d:ProgressFunction} of the progress function, in the third step we used the expression from \cref{l:outlineHParts} to rewrite $h$ in terms of $\mu := {\rm spec}(U_{T} \hat{Z}_{T} U_{T}^{\Tr} (U \hat{Z} U^{\Tr})^{-1} )$, in the fourth step we have $\mu_{i > r} = 0$ and the remaining terms can be bounded using $(\alpha-1) \mu \leq 1 + (\alpha-1) \mu$. We can bound this minimum eigenvalue as
    \begin{align*}
        \mu_{\min} = (1 + \rho_{T}(U \sqrt{\hat{Z}}))^{-1} 
        \geq \bigg( 1 + \frac{\max_{j \not\in T} \hat{z}_{j}}{\min_{j \in T} \hat{z}_{j}} \rho_{T}(U) \bigg)^{-1}
        = \bigg( 1 + \frac{\hat{z}_{k+1}}{\hat{z}_{k}} \rho_{T}(U) \bigg)^{-1}
        = \frac{1}{1 + \delta} 
        \geq 1 - \delta  , 
    \end{align*}
    where the first step was by the eigenvalue bound in \cref{f:spectralRho} applied to $V := U \hat{Z}^{1/2}$, in the second step we applied \cref{l:rhoZBound}, in the third step we used that $\hat{z}$ is in increasing order, the fourth step was by our choice of $\frac{\hat{z}_{k}}{\hat{z}_{k+1}} = \frac{\rho_{T}(U)}{\delta}$, and the final step was by Taylor approximation for $0 < \delta \leq 1$. 
    
    Putting this together gives the bound $\|\lev^{U}(z) - \lev^{U}(\hat{z})\|_{1} \leq 2d \delta$, and applying this bound for each of the $\leq n$ prefix sets gives the result. 

\end{proof}

    \begin{remark}
     In the above procedure, we only round prefix sets, so we can strengthen the $n \log \rho$ term to
    \[ \max_{T_{1} \subseteq ... \subseteq T_{n}} \sum_{k \in [n]} \log \rho_{T_{k}}(U) .   \]
    We conjecture that the guarantee in \cref{t:Regularization} can be improved to $\log \frac{\hat{z}_{\max}}{\hat{z}_{\min}} \lesssim \log(\rho/\delta)$. 
    This and further improvements we leave to future work. 

    \end{remark}    

    Actually, it is well known that $\overline{\chi}$ and $\rho$ are NP-hard to approximate even to exponential accuracy \cite{TuncelChiNP}. But note that the procedure described in the proof above only requires us to evaluate $\rho_{T}(U)$ for polynomially many sets $T \subseteq [n]$. In the following \cref{ss:RegBitComplexity}, we show how to turn the above \cref{t:Regularization} into a strongly polynomial algorithm by using simple overestimates for $\rho_{T}$.

\subsection{Bit Complexity Bound} \label{ss:RegBitComplexity}
    
    So far we have only given bounds on (relative) magnitudes of scalings and updates.
    In order to bound the bit complexity, we first show that leverage scores are robust to small multiplicative perturbations of scalings.  

    \begin{lemma} \label{l:scalingRobustness}
        For input $U \in \R^{d \times n}$ and scaling $z \in \R^{n}_{++}$, if $\hat{z} \in (1 \pm \delta) z$ for $0 \leq \delta < 1/2$, then
        \[ \lev^{U}(\hat{z}) \in (1 \pm 3\delta) \lev^{U}(z) .    \]
    \end{lemma}
    \begin{proof}
    This follows by a straightforward Taylor approximation
        \[ \lev^{U}_{j}(\hat{z}) = \hat{z}_{j} u_{j}^{\Tr} ( U \hat{Z} U^{\Tr})^{-1} u_{j} \in (1 \pm \delta) z_{j} u_{j}^{\Tr} ( (1 \pm \delta) U Z U^{\Tr})^{-1} u_{j}  
        \in (1 \pm 3 \delta) \lev^{U}_{j}(z) .  \]
        where in the second step we used $\hat{z}_{j} \in (1 \pm \delta) z_{j}$ and $\hat{Z} \in (1 \pm \delta) Z$ which implies $U \hat{Z} U^{\Tr} \in (1 \pm \delta) U Z U^{\Tr}$ and therefore $(U \hat{Z} U^{\Tr})^{-1} \in (1 \pm \delta)^{-1} (U Z U^{\Tr})$, and in the final step we used $|\frac{1+\delta}{1-\delta} - 1| \leq 3\delta$ for $0 \leq \delta < 1/2$. 
    \end{proof}

    Therefore, if the relative magnitude of our scalings remain bounded, then we can always round to sufficient accuracy to maintain our guarantees with bounded bit complexity. This is the only place in our work that requires a rounding oracle and therefore does not fit into the real model. For further discussion, see \cref{s:stronglyPoly}.  
    
    The following well-known result allows us bound $\rho$ in terms of bit complexity. 

    \begin{theorem} [Theorem 6 in \cite{VYChiBound}] \label{t:VYChiBound}
        For $U \in \R^{d \times n}$ with each entry having bit complexity $b$, 
        \[ \log \overline{\chi}(U) \leq O(d (b + \log d) + \log n) , \qquad \log \rho(U) \leq O(d (b + \log d) + \log n) .  \]
    \end{theorem}

    As a consequence, we can maintain bounded bit complexity for all intermediate scalings. We emphasize that the procedure below is the only time we require the `rounding' oracle described in \cref{s:stronglyPoly}.  
    
\begin{corollary} \label{c:RegularizationBit}
    Given any frame $U \in \R^{d \times n}$ of bit complexity $b$, for any scaling $z \in \R^{n}_{++}$ and precision $0 < \delta < 1/2$, there is $\hat{z} \in \R_{++}^{n}$ with (entry-wise) bit complexity $O(n (d b + \log n + \log(1/\delta)) )$ that satisfies $\|\lev^{U}(z) - \lev^{U}(\hat{z})\|_{1} \leq O(n d \delta)$. 
    Further, $\hat{z}$ can be computed in strongly polynomial time, involving $O(n \log n)$ comparisons and $O(n)$ matrix operations on $d \times d$ and $d \times n$ matrices (Gram-Schmidt, inverse, multiplication).  
\end{corollary}
\begin{proof}
    We follow the same procedure as in \cref{t:Regularization}, but all quantities are maintained only up to bounded precision by using the `rounding' oracle described in \cref{s:stronglyPoly}. 

    We first note that in strongly polynomial time, we can compute an overestimate for $\rho_{T}(U)$ as 
    \[ 1 + \rho_{T}(U) = \|((U U^{\Tr})^{-1/2} U_{T})^{+}\|_{\ope}^{2} \leq \|((U U^{\Tr})^{-1/2} U_{T})^{+}\|_{F}^{2} = \tr[(U_{T}^{\Tr} (U U^{\Tr})^{-1} U_{T})^{+}] \leq d (1 + \rho_{T}(U)) , \]
    to get a slightly worse guarantee in \cref{t:Regularization}: there exists $\hat{z} \in \R^{n}_{++}$ such that $\|\lev^{U}(z) - \lev^{U}(\hat{z})\|_{1} \leq 2 n d \delta$ and $\log \frac{\hat{z}_{\max}}{\hat{z}_{\min}} \leq n \log( d (1+\rho(U)) / \delta)$. 

    By dividing $\hat{z} \leftarrow \hat{z}/\hat{z}_{\min}$, we can assume without loss that $\hat{z}_{\min} = 1$ and $\log \hat{z}_{\max} \leq n \log( d (1+\rho(U)) / \delta)$. Therefore, we can use $N := \lceil n \log( d(1+ \rho(U))/\delta) + \log(1/\delta) \rceil$ bits for each entry and maintain a good approximation. Explicitly, we apply the rounding oracle described in \cref{s:stronglyPoly} to get $\bar{z} \leftarrow \lfloor z/\delta \rceil \delta$, so that $\bar{z}$ has entrywise bit complexity $\leq N$. Further, $\bar{z} \in (1 \pm \delta/2) \hat{z}$ which by \cref{l:scalingRobustness} implies $\lev^{U}(\bar{z}) \in (1 \pm O(\delta)) \lev^{U}(\hat{z})$. Since $\lev \in [0,1]^{n}$, we can combine this with the error guarantee of \cref{t:Regularization} to get 
    \[ \|\lev^{U}(z) - \lev^{U}(\bar{z})\|_{1} \leq \|\lev^{U}(z) - \lev^{U}(\hat{z})\|_{1} + \|\lev^{U}(\hat{z}) - \lev^{U}(\bar{z})\|_{1} \leq O(nd \delta) .  \]
    And applying the bound $\log(\rho(U)) \leq O(d( b + \log d) + \log n)$ from \cref{t:VYChiBound} gives the required bound on bit complexity. 
\end{proof}

\begin{remark}
    In the above procedure, we computed upper bounds for $\rho_{T}$ for sets $\{T = \{k\}\}_{k \in [n]}$. To reduce runtime, we note the following fact: 
    \[ S \subseteq T, \rk_{U}(S) = \rk_{U}(T) \implies \rho_{S}(U) \geq \rho_{T}(U) . \]
    Therefore we can perform a single Gram-Schmidt and just compute $\rho_{T}$ for the $d$ sets where the rank increases. 
\end{remark}

    \subsection{Comparison to Prior Work} \label{ss:RegPreviousWork}

    In this subsection we will compare our regularization \cref{t:Regularization} to similar results in the literature on algorithms for frame scaling. 

    The strongly polynomial algorithm of \cite{DKT} for Forster transformation (the special case of \cref{d:introForsterTransform} with $c = \frac{d}{n} 1_{n}$) maintains left scaling $L \in \R^{d \times d}$ while keeping the right scaling implicit. The regularization result given in Theorem 5.1 of \cite{DKT} is therefore much more complicated, as they need to show how to round the left scaling matrix $L \in \R^{d \times d}$ to bounded bit complexity while maintaining small error with respect to the frame scaling problem \cref{d:introForsterTransform}.

    All other results come from weakly polynomial algorithms from frame scaling. The line of works \cite{HardtMoitra}, \cite{AKS}, \cite{SinghVishnoi}, \cite{StraszakVishnoi}, \cite{IPMforGP} all study a convex formulation for frame scaling and show that, for any desired precision $\delta > 0$, there exists a scaling with magnitude bounded by a function of $1/\delta$, that achieves optimality gap $\delta$, and therefore has marginals that are $\delta$-close to the desired marginals. This is then combined with off-the-shelf optimization methods to give a weakly polynomial algorithm to solve the convex formulation. 

    In our work, we require an algorithm to produce a bounded scaling that satisfies the weaker condition that the marginals are close, instead of small optimality gap. The results given in the previous works often give an existential proof of the bound on scalings, whereas in our work we need a constructive version. Our regularization procedure, sorting the entries of the scaling and then shrinking gaps to appropriate size, can in fact be seen as a special case of the implicit procedures used in \cite{StraszakVishnoi} and \cite{IPMforGP} for the more general `geometric programming' setting.
    Our main contribution is therefore to make this procedure explicit, and to give a refined analysis of the error bound in our frame scaling setting. We believe that our relation to the refined $\rho$ parameter is of independent interest for scaling.

\section{Putting it Together} \label{s:FinalProofs}

Now that we have all the pieces, we can analyze the convergence of Algorithm \ref{a:MainAlgorithm}. 

\begin{proof} [Proof of \cref{t:mainAlgAnalysis}]
    We show that in each iteration we either get a certificate of infeasibility $\langle c, 1_{T} \rangle > \rk(U_{T})$ or make $1 - 1/O(n^{3})$ multiplicative progress in the error $\|\lev - c\|_{2}^{2}$, where $\lev := \lev^{U}$ for input frame $U$. This suffices to give the iteration bound as $\lev, c \in [0,1]^{n}$ so $\|\lev - c\|_{2}^{2} \leq n$.

    Fix iteration $t$ and let $z := z^{(t)}$. By \cref{p:outlinePolyGap} we can compute a set $T \subseteq [n]$ with margin $\gamma^{2} \geq \|\lev(z) - c\|_{2}^{2}/2 n^{3}$. For $h := h_{T}^{U,z}$ according to \cref{d:ProgressFunction}, and assuming we are in the feasible case, we can compute $\hat{\alpha} \geq 1$ satisfying $\gamma/5 \leq h(\hat{\alpha}) - h(1) \leq \gamma$ by
    \cref{t:updateGuarantee}.
    Applying \cref{l:outlineProgressLemma} to update $z' := z \circ (1_{\overline{T}} + \hat{\alpha} 1_{T})$ gives
    \[ \|\lev(z) - c\|_{2}^{2} - \|\lev(z') - c\|_{2}^{2} \geq 2 \gamma (h(\hat{\alpha}) - h(0)) \geq 2 \gamma^{2}/5 .    \]

    We show this progress lower bound holds up to a constant factor even after regularization: let $\hat{z}$ be the output of \cref{c:RegularizationBit} with parameter $\delta$ chosen later. Then we can bound
    \[ \|\lev(\hat{z}) - c\|_{2} - \|\lev(z') - c\|_{2} \leq \|\lev(\hat{z}) - \lev(z')\|_{2} \leq \|\lev(\hat{z}) - \lev(z')\|_{1} \leq 2 n d \delta ,   \]
    where the first step was by triangle inequality and the final step was by the guarantee of \cref{c:RegularizationBit}. For $\delta \in (0,1/2)$ chosen later, this implies 
    \begin{align*} 
    \|\lev(\hat{z}) - c\|_{2}^{2} - \|\lev(z') - c\|_{2}^{2} & = (\|\lev(\hat{z}) - c\|_{2} + \|\lev(z') - c\|_{2} ) (\|\lev(\hat{z}) - c\|_{2} - \|\lev(z') - c\|_{2} )
    \\ & \leq 4 (\sqrt{2 n^{3} \gamma^{2}} + nd\delta ) (n d \delta), 
    \end{align*}
    where we bounded $\|\lev(\hat{z}) - c\|_{2} - \|\lev(z') - c\|_{2} \leq 2 nd \delta$ by the calculation above and $\|\lev(z') - c\|_{2} \leq \|\lev(z) - c\|_{2} \leq \sqrt{2 n^{3} \gamma^{2}}$, as our update $z'$ decreases the error as shown above, and the second inequality is by \cref{p:outlinePolyGap} for margin $\gamma$. Choosing $\delta = \gamma/(15 n^{5/2} d)$, we can bound this whole expression by $\gamma^{2}/5$. In total this gives 
    \begin{align*} 
    \|\lev(z) - c\|_{2}^{2} - \|\lev(\hat{z}) - c\|_{2}^{2} & = (\|\lev(z) - c\|_{2}^{2} - \|\lev(z') - c\|_{2}^{2}) - (\|\lev(\hat{z}) - c\|_{2}^{2} - \|\lev(z') - c\|_{2}^{2}) 
    \geq \frac{2 \gamma^{2}}{5} - \frac{\gamma^{2}}{5}  , 
    \end{align*}
    where in the second step we lower bounded the progress term by $2 \gamma^{2}/5$ and upper bounded the regularization term by $\gamma^{2}/5$ as shown above. By \cref{p:outlinePolyGap}, we have $\gamma^{2} \geq \|\lev(z) - c\|_{2}^{2}/2n$, so this implies that in $O(n^{3} \log(n/\eps))$ iterations we have 
    \[ \|\lev(z^{(t)}) - c\|_{2}^{2} \leq \|\lev - c\|_{2}^{2} \bigg( 1 - \frac{1}{O(n^{3})} \bigg)^{O(n^{3} \log(n/\eps))} \leq n \cdot e^{- O(\log (n/\eps))} \leq \eps^{2} ,    \]
    where we bounded $\|\lev - c\|_{2}^{2} \leq n$ in the second step as $\lev, c \in [0,1]^{n}$. 

    Next, we bound the number of operations in each iteration. We can find the set $T$ using $O(n)$ matrix multiplications and inversions of $d \times d$ matrices, along with $O(n \log n)$ comparison operations. Verifying feasibility requires a rank computation, which can be performed using Gram-Schmidt applied to a $d \times n$ matrix. For the Newton-Dinkelbach iterations we compute $h$ and $h'$ using the explicit expressions described in \cref{r:hStronglyPoly}, which can be done using $O(1)$ matrix multiplications and inversions involving $d \times d$ matrices. 

    For the guessing Algorithm \ref{a:GuessAlgorithm}, we require $O(nd + n d^{2} \log n)$ determinants, inverses, and multiplications with $d \times d$ matrices. 
    For the regularization procedure in \cref{c:RegularizationBit}, we sort the entries in $O(n \log n)$ time and compute $\rho_{T}(U)$ for $n$ sets, each of which requires a Gram-Schmidt on a $d \times n$ matrix, along with $O(1)$ matrix multiplications and inversions of $d \times d$ matrices. 
    
    Finally, assuming that the bit complexity of the input $U$ is $b$, the regularization step guarantees that the (entry-wise) bit complexity of the scaling at the end of each iteration remains bound by 
    \[  n \log(\rho/\delta) \leq O(n (d b + \log(nd/\eps))) ,    \]
    by \cref{c:RegularizationBit} and our choice $\delta = \Omega(\gamma/\poly(nd))$ applied above for $\gamma \geq \Omega(\eps)$. Further, \cref{p:iterateUB} implies that the bit complexity of the scaling after the update step, before regularization, is also bounded by the same quantity up to constant factor. 
    
    Therefore the algorithm is strongly polynomial with given operation counts. 

    \end{proof}

\section{Improvement of Matrix Scaling} \label{s:matrixImprovement}

In this section we show how our techniques give a faster strongly polynomial algorithm for the matrix scaling problem. 
In this setting, we are given $A \in \R^{m \times n}_{+}$ and want to find diagonal scalings $L \in \diag_{++}(m), R \in \diag_{++}(n)$ such that $B := L A R$ satisfies $B 1_{n} = r, B^{\Tr} 1_{m} = c$. Formally: 

\begin{definition} [Matrix Scaling Problem] \label{d:matrixScalingInput}
    For input $A \in \R^{m \times n}_{+}$, let $(r(A), c(A)) \in \R^{m} \times \R^{n}$ be the row and column sums. Explicitly, 
\[ r_{i}(A) := \sum_{j \in [n]} A_{ij}, \qquad c_{j}(A) := \sum_{i \in [m]} A_{ij}   . \]
Given desired marginals $(r, c) \in \R^{m}_{+} \times \R^{n}_{+}$ and precision $\eps > 0$, find diagonal scalings $L \in \diag_{++}(m), R \in \diag_{++}(n)$ such that $B := L A R$ is $\eps$-approximately $(r, c)$-scaled: 
\[ \|r(B) - r\|_{2}^{2} + \|c(B) - c\|_{2}^{2} \leq \eps^{2} .     \]
\end{definition}

We will also use the combinatorial structure of the input $A$ (i.e. ${\rm supp}(A)$) in our algorithms. 

\begin{definition} \label{d:matrixNbrhood}
    For input $A \in \R^{m \times n}_{+}$ and column subset $T \subseteq [n]$, the neighborhood is 
    \[ N(T) := \{i \in [m] \mid \exists j \in T: A_{ij} > 0 \} .   \]
\end{definition}

If we interpret $A \in \R^{m \times n}_{+}$ as the incidence matrix of a weighted bipartite graph, then $N(T)$ corresponds to the neighborhood in this graph. 

In \cite{LSW}, the authors give the first strongly polynomial algorithm for this problem. 
By adapting our frame scaling analysis to this setting, we are able to reduce the number of iterations from $O(n^{5} \log(n/\eps))$ as given in \cite{LSW} to $O(n^{3} \log(n/\eps))$ iterations, matching our Main \cref{t:mainAlgAnalysis}. Further, each iteration can still be performed in nearly linear time, so this gives a quadratic improvement in runtime over the algorithm in \cite{LSW}. 
In order to deal with the polynomial bit complexity requirement, the work of \cite{LSW} does not use a rounding step, but instead maintains quantities in floating point representation and argues that the exponent remains of polynomial bit size. 

We show that we can also adapt our regularization procedure from \cref{s:regularization} to the matrix setting, giving a new stronger bound on the size and bit complexity of scalings required for $\eps$-approximate matrix scaling. 
This also allows us to avoid the complications that arise from merging floating point precision with the real model as discussed in \cref{s:stronglyPoly}. 

Our main result is as follows:

\begin{theorem} \label{t:mainMatrixScaling}
    Given input $A \in \R^{m \times n}_{+}$ with desired marginals $(r,c) \in \R^{m}_{++} \times \R^{n}_{++}$ satisfying $s := \langle r, 1_{m} \rangle = \langle c, 1_{n} \rangle$, and precision $\eps > 0$, there is a strongly polynomial algorithm to produce $\eps$-approximate scalings or a certificate of infeasibility. This algorithm takes $O(n^{3} \log(sn/\eps))$ iterations, each of which require $\tilde{O}(\text{nnz}(A))$ arithmetic operations (here $\text{nnz}(A)$ is the number of non-zero entries of $A$). Further, if the input $A$ has entrywise bit complexity $\leq b$, then the intermediate scalings all have bit complexity bounded by $\poly(n,m,b) \log(1/\eps)$.  
\end{theorem}

    Matrix scaling is a fundamental problem that has appeared in a variety of fields throughout science and engineering (for a thorough survey, see \cite{Idel}). Recently, there has been a renewed interest in fast algorithms for this problem. In particular, \cite{CMTV17} and \cite{AZLOW} study a convex formulation for matrix scaling and use sophisticated optimization techniques to solve it in near-linear time. 
    While this gives fast runtime in terms of the input dimensions, the number of iterations of these algorithms depends on quantities such as $A_{\min} := \min \{ A_{ij} \mid (i,j) \in {\rm supp}(A) \}$, and so these are weakly polynomial algorithms.

In the remainder of the section, we discuss how to adapt our algorithm and analysis to the matrix setting. 
We first recall the high-level outline for frame scaling: (1) reduce to just the right scaling; (2) use the norm of the column error as a potential function; (3) find a set with large `margin' to scale up; (4) reduce the analysis of the potential decrease to a proxy function; (5) find a good update making sufficient progress for the proxy function; (6) finally, apply regularization procedure to control the size and bit complexity of intermediate scalings. 
The matrix case will follow same outline, so we mainly focus on the necessary modifications and omit most proofs.

We begin by reducing to row-scaled matrices. Specifically, for any column scaling $y \in \R^{n}_{++}$ we implicitly match the desired row sums, namely using row scaling $x_{i} := r_{i} / \sum_{j \in [n]} A_{ij} y_{j}$. 
This allows us to just keep track of column scalings $y$ and the column sums of $X A Y$ induced by this implicit row scaling. 
 
\begin{definition} [Matrix version of \cref{d:outlineLeverageScores}] \label{d:matrixCols}
    For matrix scaling input $(A,r,c)$ as in \cref{t:mainMatrixScaling}, for any right scaling $y \in \R^{n}_{++}$, 
    \[ c^{A,r}_{j}(y) := \sum_{i \in [m]} r_{i} \frac{A_{ij} y_{j}}{\sum_{k \in [n]} A_{ik} y_{k}} .    \]
\end{definition}

We use potential function $\|c^{A,r}(y) - c\|_{2}^{2}$ and attempt to compute an update in each step that decreases this by a multiplicative factor. 
We follow the same update strategy, computing set $T \subseteq [n]$ with largest margin, where we recall the margin of $T \subseteq [n]$ is defined as the largest $\gamma \geq 0$ such that there exists a threshold $\nu \in \R$ satisfying 
    \[ \max_{j \in T} (c^{A,r}_{j}(y) - c_{j}) \leq \nu - \gamma \leq \nu + \gamma \leq \min_{j \not\in T} (c^{A,r}_{j}(y) - c_{j}) .    \]
    This allows us to apply the lower bound from \cref{p:outlinePolyGap} verbatim: 

\begin{prop} [Matrix version of \cref{p:outlinePolyGap}] \label{p:matrixPolyGap}
For input $A \in \R^{m \times n}_{+}$, desired marginals $(r,c) \in \R^{m+n}_{+}$, and current scaling $y \in \R^{n}_{++}$, sort $c^{A,r}(y) - c$ in non-decreasing order, and let $T \subseteq [n]$ be the prefix set with maximum margin $\gamma$. Then
\[ \gamma^{2} \geq \frac{1}{2n^{3}} \|c^{A,r}(y) - c\|_{2}^{2} .  \]
\end{prop}

    Now we want to show that we can make $\Omega(\gamma^{2})$ progress by scaling up $T$. For this purpose, we define a proxy function as follows.

    \begin{definition} [Matrix version of \cref{d:ProgressFunction}] \label{d:matrixProgressFunction}
    In the setting of \cref{p:matrixPolyGap}, the proxy function is defined as
    \begin{align*} 
    h^{A,r,y}_{T}(\alpha) & := \sum_{j \in T} c^{A,r}_{j}(y \circ (1_{\overline{T}} + \alpha 1_{T})) = \sum_{j \in T} \sum_{i \in [m]} r_{i} \frac{\alpha A_{ij} y_{j}}{\sum_{k \not\in T} A_{ik} y_{k} + \alpha \sum_{k \in T} A_{ik}} 
    = \sum_{i \in N(T)} r_{i} \frac{\alpha \mu_{i}}{1 + (\alpha-1) \mu_{i}} , 
    \end{align*}
    where $N(T)$ is the neighborhood according to \cref{d:matrixNbrhood} and $\mu_{i} := \frac{\sum_{j \in T} A_{ij} y_{j}}{\sum_{j \in [n]} A_{ij} y_{j}}$. 
    \end{definition}

    Note this is also positive sum of terms $\alpha \mu / (1 + (\alpha-1) \mu)$ for $\mu \in (0,1)$, so many of the desirable properties in our previous analysis carry to this new proxy function. In particular, $c_{j \in T}$ is increasing in $\alpha$ and $c_{j \not\in T}$ is decreasing, so we can apply the same potential analysis in \cref{l:outlineProgressLemma}: 

\begin{lemma} [Matrix version of \cref{l:outlineProgressLemma}] \label{l:matrixProgressLemma}
    Consider proxy function $h := h^{A,r,y}_{T}$ according to \cref{d:matrixProgressFunction}, where subset $T \subseteq [n]$ has margin $\gamma$. 
    Then for any $\alpha$ satisfying $0 \leq h(\alpha) - h(1) \leq \gamma$, the scaling $y' := y \circ (1_{\overline{T}} + \alpha 1_{T})$ results in 
    \[ \|c^{A,r}(y) - c\|_{2}^{2} - \|c^{A,r}(y') - c\|_{2}^{2} \geq 2 \gamma (h(\alpha) - h(1)) .  \]
\end{lemma}

 Our goal is now to find an update such that $\gamma \geq h(\alpha) - h(1) \gtrsim \gamma$, which will give the required iteration bound. 
 In order to guarantee that such an update exists, we need to use the fact that the input $A$ can be scaled to $(r,c)$. 
    Feasibility of scalings has the following characterization: 

\begin{theorem} [Prop 2.2 in \cite{LSW}] \label{t:matrixScalingFeasibility}
	For matrix scaling input $(A,r,c)$ as in \cref{d:matrixScalingInput}, the following are equivalent:
	\begin{itemize}
		\item for every $\eps > 0$, there is $L_{\eps} \in \text{diag}(m), R_{\eps} \in \text{diag}(n)$ such that $B := L_{\eps} A R_{\eps}$ is $\eps$-approximately $(r, c)$-scaled according to \cref{d:matrixScalingInput};
		\item $\langle c, 1_{n} \rangle = \langle r, 1_{m} \rangle$ and for every $T \subseteq [n]$, $c(T) := \sum_{j \in T} c_{j} \leq \sum_{i \in N(T)} r_{i} =: r(N(T))$ where $N(T)$ is the neighborhood according to \cref{d:matrixNbrhood}. 
	\end{itemize}
	$(A,r,c)$ satisfying the above conditions is called feasible, and otherwise it is infeasible and $T \subseteq [n]$ satisfying $c(T) > r(N(T))$ is a certificate of infeasibility. 
\end{theorem}

    Therefore, we can use a similar argument as in the frame case to guarantee that, for a feasible input, it is possible to find an update $\alpha$ making sufficient progress. 

\begin{prop} [Matrix version of \cref{p:outlineInfeasible}] \label{p:matrixInfeasible}
    Consider progress function $h := h^{A,r,y}_{T}$ according to \cref{d:matrixProgressFunction}, where $T \subseteq [n]$ has margin $\gamma >0$. If $r(N(T)) \geq c(T)$, then for any $\delta \in [0,\gamma)$, there is a finite solution $\alpha < \infty$ to the equation $h(\alpha) = h(1) + \delta$.
\end{prop}
\begin{proof}
    As $h$ is increasing, we compute the limit of the progress function: 
    \[ \lim_{\alpha \to \infty} h(\alpha) = \lim_{\alpha \to \infty} \sum_{i \in N(T)} r_{i} \frac{\alpha \mu_{i}}{1 + (\alpha-1) \mu_{i}} = r(N(T))  . \]
     The result then follows from exactly the same calculations as in the frame case, replacing $\lev$ with $c^{A,r}$ and using feasibility $r(N(T)) \geq c(T)$. 
\end{proof}

    Therefore in the feasible case, there is always a good update $\gamma \geq h(\alpha) - h(1) \gtrsim \gamma$, so by \cref{l:matrixProgressLemma}, we have the required iteration bound. In fact, this is the main source of our quadratic improvement in iterations: in \cite{LSW}, they use a different analysis to guarantee some entry must make $\Omega(\gamma^{2}/n^{2})$ progress. By using our proxy function, we are able to measure progress of the whole subset $T \subseteq [n]$ and guarantee $\Omega(\gamma^{2})$ progress. 

    All that remains is to actually compute an update making good progress. In the matrix setting, this is a much simpler task as we have access to $\mu$. Specifically, we use the following simple approximation
    \[ \forall \alpha \geq 1, \mu \in [0,1]: \quad \frac{\alpha \mu}{2 \max\{1, (\alpha-1)\mu\}} \leq  \frac{\alpha \mu}{1 + (\alpha-1) \mu} \leq \frac{\alpha \mu}{\max\{1, (\alpha-1)\mu\}}  .   \]
    Therefore, we have explicit access to the following piece-wise linear approximation satisfying $g(\alpha)/2 \leq h(\alpha) - h(1) \leq g(\alpha)$: 
    \[ g(\alpha) := \begin{cases}
        0 + (\alpha - 1) \sum_{i \in N(T)} r_{i} \mu_{i} (1-\mu_{i}) & \qquad \text{for} \quad \alpha - 1 \in [ 0, \frac{1}{\mu_{1}} ]
        \\ \sum_{i \leq k} r_{i} (1-\mu_{i}) + (\alpha - 1) \sum_{i > k} \mu_{i} r_{i} (1-\mu_{i}) & \qquad \text{for} \quad \alpha - 1 \in [ \frac{1}{\mu_{k}}, \frac{1}{\mu_{k+1}} ] 
        \\ \sum_{i \in N(T)} r_{i} (1-\mu_{i}) + 0 &  \qquad \text{for} \quad \alpha - 1 \in  [ \frac{1}{\mu_{|N(T)|}}, \infty ] 
    \end{cases} ,  \]
    where we sorted $\mu_{1} \geq ... \geq \mu_{|N(T)|} > 0$ and $\mu_{i > |N(T)|} := 0$. 

    Note that $g$ is increasing, so we can find the update by computing $g$ at the breakpoints $\{g(1/\mu_{k}) \}_{k \in [r]}$, finding the segment containing the solution $g(\hat{\alpha}) = \gamma$, which must exist by \cref{p:matrixInfeasible}, and then solving the linear equation in that segment. With some care, this can be performed in $\tilde{O}(\text{nnz}(A))$ time: first sort $\mu$, and then note that the linear equation for the $k$-th and $(k+1)$-st piece differs by $O(1)$ terms, so we compute them in order. This gives
    \[ \gamma/2 = g(\hat{\alpha})/2 \leq h(\hat{\alpha}) - h(1) \leq g(\hat{\alpha}) = \gamma  . \]

    Combined with the potential analysis in \cref{l:matrixProgressLemma} and gap lower bound in \cref{p:matrixPolyGap} gives 
    \[ \|c^{A,r}(y^{(t)}) - c\|_{2}^{2} - \|c^{A,r}(y^{(t+1)}) - c\|_{2}^{2} \geq 2 \gamma (h(\alpha) - h(1)) \geq \gamma^{2} \geq \|c^{A,r}(y^{(t)}) - c\|_{2}^{2}/(2 n^{3}) .   \]

    Finally, we can adapt our regularization argument to keep the magnitude and bit complexity of scalings bounded. The appropriate condition number for matrix scaling is much simpler than in the frame case. 

    \begin{definition} [Matrix version of \cref{d:rhoMeasure}] \label{d:matrixRho}
        For matrix input $A \in \R^{m \times n}_+$, 
        \[ \rho_{T}(A) := \max_{i \in N(T)} \frac{\sum_{j \not\in T} A_{ij}}{\sum_{j \in T} A_{ij}} , \qquad \rho(A) := \max_{T \subseteq [n]} \rho_{T}(A) , \]
        where $N(T)$ denotes the neighborhood of $T$ according to $A$ as discussed above. 
    \end{definition}

    Note that $\rho$ can be simply bounded by $\rho(A) \leq n \frac{A_{\max}}{A_{\min}}$ where $A_{\max,\min}$ denote the largest and smallest non-zero entries, and this is clearly bounded by $n 2^{2b}$ if all entries of $A$ have bit complexity $b$. This already gives a bound on how much the scalings grow in each iteration, as the update is bounded by $1/\mu_{\min}$ where 
    \[ \mu_{\min} := \min_{i \in N(T)} \mu_{i} = \frac{\sum_{j \in T} A_{ij} y_{j}}{\sum_{j \in [n]} A_{ij} y_{j}} \geq \frac{y_{\min}}{y_{\max}} \frac{\sum_{j \in T} A_{ij}}{\sum_{j \in [n]} A_{ij}} \geq \frac{y_{\min}}{y_{\max}} (1 + \rho_{T}(A))^{-1} .  \]

    Following the same regularization procedure in the matrix setting gives the following guarantees:
        
    \begin{corollary} [Matrix version of \cref{t:Regularization} and \cref{c:RegularizationBit}]
    Given matrix $A \in \R^{m \times n}$, for row sums $r \in \R^{m}_{+}$, scaling $y \in \R^{n}_{++}$, and precision $0 < \delta < 1/2$, there is $\hat{y} \in \R_{++}^{n}$ satisfying \[ \log \frac{\hat{y}_{\max}}{\hat{y}_{\min}} \leq n \log(\rho(A)/\delta), \qquad \text{and} \quad \|c^{A,r}(y) - c^{A,r}(\hat{y})\|_{1} \leq 2 n \delta \langle r, 1_m \rangle. \]
    Further, $\hat{y}$ can be computed in strongly polynomial time using $\tilde{O}(\text{nnz}(A))$ arithmetic operations, and if the entries of $A$ have bit complexity $b$, then the bit complexity of $\hat{y}$ is $\tilde{O}(n b)$. 
    \end{corollary}
    \begin{proof}
        The algorithm is exactly the same, sorting by $y_{1} \geq ... \geq y_{n}$ and shrinking any prefix gap. So assume there is a prefix set $T = [k]$ with index $k$ such that $\frac{y_{k}}{y_{k+1}} > \frac{\rho_{[k]}(A)}{\delta}$, where $\rho$ is according to \cref{d:matrixRho}. We scale this subset down to     
        \[ \forall j \leq k: \qquad \hat{y}_{j} := y_{j} \bigg( \frac{y_{k}}{y_{k+1}} \bigg)^{-1} \bigg( \frac{\rho_{T}(A)}{\delta} \bigg)    \]
    and $\hat{y}|_{\overline{T}} = y|_{\overline{T}}$. We once again note that $y = \hat{y} \circ (1_{\overline{T}} + \alpha 1_{T})$ for $\alpha := \frac{y_{k}}{y_{k+1}} \frac{\delta}{\rho_{T}(A)} > 1$, so we can analyze the change in column sums using progress function $h := h_{A,\hat{y}}^{T}$ according to \cref{d:matrixProgressFunction} as:
        \[ \|c^{A,r}(y) - c^{A,r}(\hat{y})\|_{1} \leq 2 (h(\alpha) - h(1)) = 2 \sum_{i \in N(T)} r_{i} \frac{(\alpha-1) \mu_{i}(1-\mu_{i})}{1 + (\alpha-1) \mu_{i})} \leq 2 \sum_{i \in N(T)} r_{i}(1 - \mu_{i}) .    \]
        Finally, we can bound $1-\mu_{\min} \geq 1-\delta$ using the large gap $\frac{y_{k}}{y_{k+1}} > \frac{\rho_{[k]}(A)}{\delta}$ and \cref{d:matrixRho} of $\rho$. Thus $2 \sum_{i \in N(T)} r_i(1-\mu_i) \leq 2 \delta \langle r, 1_m \rangle$. Since we perform this shrinking operation at most $n$ times, once for each prefix, the final error bound follows by the triangle inequality.

        It remains to justify the nearly linear running time. Note we only need to compute upper bounds for $\rho_{T}$ for prefix sets $T_{k} := [k]$ for $k \in [n]$. In this setting $N(T_{k})$ is a chain, i.e. $N([1]) \subseteq ... \subseteq N([n])$, so we can upper bound the term $\frac{\sum_{j \in T} A_{ij}}{\sum_{j \not\in T} A_{ij}}$ for each $i \in [d]$ individually by considering the first time $i \in N(T)$. 
        Explicitly, for each row $i$ in this ordering compute $\sum_{j \in T} A_{ij}$ for first time $i \in N(T)$, and for each $T$ just consider these as upper bounds for $\rho_{T} := \max_{i \in N(T)} \frac{\sum_{j \in T} A_{ij}}{\sum_{j \not\in T} A_{ij}}$.         
    \end{proof}

    Therefore, by regularizing to precision $\delta = \poly(\gamma/s n)$ for $s := \langle r, 1_{m} \rangle = \langle c, 1_{n} \rangle$, we still make $\Omega(\gamma^{2}) \geq \Omega(\|c^{A,r}(y) - c\|_{2}^{2}/n^{3})$ progress in each iteration. Noting the original error $\|c^{A,r}(1_{n}) - c\|_{2}^{2} \leq n s^{2}$, this gives the required bound on the number of iterations.
    And by applying the regularization procedure with bounded bit precision, we maintain $\poly(s,n,b) \log(1/\eps)$ bit complexity for all intermediate scalings.

\section{Strongly Polynomial Models of Computation} \label{s:stronglyPoly}

In this section, we discuss the precise sense in which our algorithm is strongly polynomial. 
We will also place this model in context with other commonly used computational models and give motivation for its use. Some of the below descriptions are inspired by the recent survey of Srivastava \cite{SrivastavaSurvey} on the complexity of eigendecomposition. 

\begin{enumerate}
    \item \textbf{Real Model}: the input is given as $n$ real numbers; the algorithm is allowed oracle access to exact arithmetic gates, which take constant time; here, a polynomial time algorithm is one that makes $\poly(n)$ calls to the arithmetic oracle. This model is studied in complexity theory and pure mathematics. 

    \item \textbf{Bit model}: the input is given as $n$ numbers each consisting of $b$ bits; the algorithm is able to perform exact arithmetic operations, but the cost is the number of bit operations; since bit size can grow with each operations, the algorithm is allowed the freedom to round arbitrarily; here, a polynomial time algorithm performs $\poly(n,b)$ bit operations. Note that this automatically enforces a $\poly(n,b)$ space constraint on all intermediate quantities. This is the most prominent computational model used to define tractable problems from the perspective of theoretical computer science. 

    \item \textbf{Finite precision model}: the input is given as $n$ numbers in floating point with fixed machine precision $b$ bits; the algorithm is allowed arithmetic operations but they are only performed approximately up to adversarial multiplicative error that depends on the machine precision; here once again a polynomial time algorithm performs $\poly(n,b)$ bit operations, but because all intermediate quantities remain at fixed precision, this is equivalent to performing $\poly(n,b)$ arithmetic operations. This is the dominant model used in numerical analysis.  
\end{enumerate}

As described in \cref{ss:stronglyPoly}, the strongly polynomial model is a blend of the real and bit models: arithmetic operations are performed using exact arithmetic and only take constant time, but the bit complexity of intermediate quantities still needs to remain polynomially bounded in the input complexity.  

The prototypical strongly polynomial algorithm in linear algebra is the computation of the determinant. 
In practice, this is performed in the finite precision model by applying Gaussian elimination to reduce to row-echelon form and computing an approximation of the determinant that depends on machine precision. But this naive algorithm is not strongly polynomial, as it is well-known that Gaussian elimination with exact arithmetic and arbitrary pivots could lead to an exponential size blow up in the bit complexity \cite{GaussianEliminationWorstCase}. 
In \cite{Edmonds}, Edmonds gave a simple Schur-complement style algorithm that computes the determinant in the real model. Furthermore, all intermediate quantities are $k \times k$ subdeterminants of the original matrix for $k \in [n]$, so have polynomial bit complexity in the original input. This algorithm therefore satisfies all strongly polynomial requirements. 

Smale \cite{Smale} famously conjectured the existence of strongly polynomial algorithms for general linear programs. Several such algorithms have been developed for special classes of LP's, such as flow problems on graphs, but the general case is still open. In \cite{ScalingInvIPM}, the authors build an interior point framework for LP whose complexity depends only on the constraint matrix. This gives a unified strongly polynomial algorithm for many known important cases of LP. 

The algorithms we consider in this paper almost fit into this model, but require slightly extra power. Specifically, we consider the stronger version of strong polynomial as defined in section 1.3 of \cite{GLS} where we are given an additional `rounding oracle': given input $x \in \R, b \in \mathbb{N}$, output $x$ rounded to the nearest multiple of $2^{-b}$. 
Note that this is not oblivious to the bit representation of the input, and so does not fit into the real model interpretation of strongly polynomial. We require this extra power only in our regularization step in \cref{s:regularization}, and otherwise we only require exact arithmetic operations. This difficulty is encountered in many other other strongly polynomial algorithms, such as the original matrix scaling algorithm of \cite{LSW}. 
We believe it may be possible to avoid this rounding step by maintaining scalings as small complexity combinations of some `canonical scalings', and we leave this as an interesting open problem. 

\section{Applications to Learning Theory} \label{s:learningTheory}

In this section we discuss two previous works that apply frame scaling in the context of machine learning. In particular, we will show how our new algorithm for frame scaling can be used to improve the runtime of the learning algorithms of \cite{DKT} and \cite{PointLocation}. 

\subsection{Halfspace Learning \cite{DKT}}

We first discuss the application of frame scaling to halfspace learning.
We begin by formally stating the problem. In \cite{DKT}, the goal is to learn an unknown affine halfspace defined by normal vector $w \in \R^{d}$. 

The algorithm gets labeled samples from an unknown distribution $(x,y) \sim \mathcal{D}$, where $x \in \R^{d}, y = \text{sign}(\langle w, x \rangle)$. And the main result is as follows:

\begin{theorem} [Theorem 1.6 in \cite{DKT}] \label{t:DKTLearning}
    Given unknown halfspace defined by normal vector $w \in \R^{d+1}$ and unknown distribution $\mathcal{D}$ as above, there is a strongly polynomial algorithm that, for any $\delta > 0$, takes $\poly(d/\delta)$ samples from $\mathcal{D}$ and with high probability outputs a hypothesis $f : \R^{d} \to \{\pm 1\}$ satisfying
    \[ Pr_{(x,y) \sim \mathcal{D}}[f(x) \neq y] \leq \delta.     \]
\end{theorem}

They are also able to extend this result to the inhomogeneous setting, where the halfspace has threshold $\{x \mid \langle w, x \rangle \geq t\}$, as well certain noisy models, where the labels only agree with the halfspace on some large fraction of samples (see Theorem 1.8 in \cite{DKT}). We focus on showing that our new frame scaling algorithm can also be applied in the simpler homogeneous setting with exact labels, noting that the ideas can be extended in a straightforward way. 

Frame scaling is helpful in this setting for two key reasons: (1) points with large ``margin" can be easily classified; (2) for point sets in approximate $(I_{d}, c = \frac{d}{n} 1_{n})$-position, a large fraction of points have large ``margin". These two observations are formalized in Lemma 7.2 and 7.3 in \cite{DKT} respectively. We show that these technical lemmas can be appropriately modified to work with our implicit frame scaling representation.

We first show that the large ``margin" classifier can be modified to work in strongly polynomial time for arbitrary inner product. We will apply this for inner product $\langle u, v \rangle_{Q} := u^{\Tr} (U Z U^{\Tr})^{-1} v$ for our frame $U \in \R^{d \times n}$ and scaling $z \in \R^{n}$. 
This follows the ``Improved Perceptron" algorithm given in Dunagan and Vempala \cite{DV04}.

\begin{lemma} [Lemma 7.2 in \cite{DKT} (Perceptron)] \label{l:improvedPerceptron}
    Consider $n$ labeled examples $\{(u_{j},y_{j}) \in \R^{d} \times \{\pm 1\}\}_{j \in [n]}$, such that there is an unknown linear halfspace satisfying $\forall j \in [n]: y_{j} = \text{sign}(\langle w, u_{j} \rangle_{Q})$ for some inner product $\langle \cdot, \cdot \rangle_{Q}$. 
    Then for any $\gamma \in (0,1)$ and initial guess $v_{0}$ satisfying $\langle v_{0}, w \rangle_{Q} > 0$, there is an algorithm requiring $\log_{1-\gamma^{2}}(\frac{\langle v_{0}, w \rangle_{Q}^{2}}{\|v_{0}\|_{Q}^{2} \|w\|_{Q}^{2}})$ iterations that outputs guess $v_{T}$ such that 
    \[ \langle v_{T}, u_{j} \rangle_{Q}^{2} \geq \gamma^{2} \|v_{T}\|_{Q}^{2} \|u_{j}\|_{Q}^{2} \implies \text{sign}(\langle v_{T}, u_{j} \rangle) = y_{j} .    \]    
\end{lemma}
\begin{proof}
    The algorithm is as follows: in each iteration, if there is some $j \in [n]$ such that $v$ mis-classifies $u_{j}$ and there is a large margin, then update
\[ y_{j} \langle v_{t}, u_{j} \rangle_{Q} \leq - \gamma \|v_{t}\|_{Q} \|u_{j}\|_{Q}  \implies  \quad v_{t+1} := v_{t} - \frac{\langle v_{t}, u_{j} \rangle_{Q}}{\|u_{j}\|_{Q}^{2}} u_{j} .    \]

We will use the correlation $\frac{\langle v, w \rangle}{\|v\|_{Q} \|w\|_{Q}}$ as a potential function and show that it improves by a geometric factor after each update. This will imply the iteration bound. 

For shorthand, let $v$ be our starting vector, and $v'$ be the vector after updating the $j$-th sample. 
First note that, since $w$ correctly classifies all the samples, we have 
\[ \langle v', w \rangle_{Q} - \langle v, w \rangle_{Q} = \frac{- \langle v, u_{j} \rangle_{Q} 
\langle u_{j}, w \rangle_{Q}}{\|u_{j}\|_{Q}^{2}} > 0 ,    \]
where in the first step we used the definition of the update, and the final inequality was because $v$ and $w$ disagree on the classification of the $j$-th sample: $\text{sign}(\langle w, u_{j} \rangle_{Q}) \neq \text{sign}(\langle v, u_{j} \rangle_{Q})$.

Next, we can show that the norm of our guess is decreasing geometrically:
\[ \|v'\|_{Q}^{2} = \|v\|_{Q}^{2} - 2 \langle v, u_{j} \rangle_{Q} \frac{\langle v, u_{j} \rangle_{Q}}{\|u_{j}\|_{Q}^{2}} + \|u_{j}\|_{Q}^{2} \frac{\langle v, u_{j} \rangle_{Q}^{2}}{\|u_{j}\|_{Q}^{4}} = \|v\|_{Q}^{2} - \frac{\langle v, u_{j} \rangle_{Q}^{2}}{\|u_{j}\|_{Q}^{2}} \leq (1-\gamma^{2}) \|v\|_{Q}^{2} ,   \]
where the final step was by the margin assumption. 

Therefore, the correlation improves after each update:
\[ \frac{\langle v', w \rangle_{Q}}{\|v'\|_{Q} \|w\|_{Q}} \geq \frac{1}{\sqrt{1-\gamma^{2}}} \frac{\langle v, w \rangle_{Q}}{\|v\|_{Q} \|w\|_{Q}} ,   \]
where we used $\langle v', w \rangle > \langle v, w \rangle_{Q}$ for the numerator and $\|v'\|_{Q}^{2} \leq (1-\gamma^{2}) \|v\|_{Q}^{2}$. Since we have initial correlation $\langle v_{0}, w \rangle_{Q} > 0$ and the maximum is $\leq 1$ by Cauchy-Schwarz, this gives the required iteration bound. 
\end{proof}

Next, if the points are in approximate radial isotropic position, then a large fraction of the points will have large margin. 

\begin{lemma} [Lemma 7.3 in \cite{DKT}] \label{l:ForsterMargin}
    Consider frame $U \in \R^{d \times n}$ with scaling $z \in \R^{n}_{++}$ is in $\eps$-approximate $(I_{d}, c = \frac{d}{n} 1_{n})$-position according to \cref{d:introForsterTransform} with $\eps \leq \frac{d}{2n}$, and inner product $\langle u, v \rangle_{Q} := u^{\Tr} (U Z U^{\Tr})^{-1} v$. Then for any $w \in \R^{d}$,
    \[ T := \bigg\{ j \in [n] :  \frac{\langle w, u_{j} \rangle_{Q}^{2}}{\|w\|_{Q}^{2} \|u_{j}\|_{Q}^{2}} \geq \frac{1}{4d} \bigg\} \implies \quad \frac{|T|}{n} \geq \frac{1}{5d} .    \]
\end{lemma}
\begin{proof}
    By the approximate $(I_{d},c)$-position, we have $\forall j \in [n]: \lev^{U}_{j}(z) \in (1 \pm 1/2) \frac{d}{n}$ as $\|\lev^{U}(z) - \frac{d}{n} 1_{n}\|_{2}^{2} \leq (\frac{d}{2n})^{2}$ by assumption. Then, following the argument in Lemma 7.3 of \cite{DKT}, for unknown $w \in \R^{d}$ we have
\[ \sum_{j \in [n]} z_{j} \langle w, u_{j} \rangle_{Q}^{2} = w^{\Tr} (U Z U^{\Tr})^{-1} (U Z U^{\Tr}) (U Z U^{\Tr})^{-1} w = \|w\|_{Q}^{2} .   \]
Also, by the approximate radial isotropy condition and Cauchy-Schwarz, we have 
\[ \forall j \in [n]:  \frac{z_{j} \langle w, u_{j} \rangle_{Q}^{2}}{\|w\|_{Q}^{2}} \leq z_{j} \|u_{j}\|_{Q}^{2} = \lev^{U}_{j}(z) \leq \frac{3}{2} \frac{d}{n} .  \]
Therefore by a simple averaging argument, there must be $\Omega(1/d)$ points with large margin: 
\begin{align*}
    1 = \sum_{j \in [n]} \frac{z_{j} \langle w, u_{j} \rangle_{Q}^{2}}{\|w\|_{Q}^{2}}
    < \sum_{j \in T} z_{j} \|u_{j}\|_{Q}^{2} + \sum_{j \not\in T} \frac{z_{j} \|u_{j}\|_{Q}^{2}}{4d} \leq |T| \frac{3}{2} \frac{d}{n} + (n - |T|) \frac{3}{2} \frac{1}{4n} , 
\end{align*}
where the first step was by the isotropy assumption and calculation above, in the second step we used the Cauchy-Schwarz bound for $j \in T$ and the low margin assumption for $j \not\in T$, and in the final step we used the bound on the norms. Rearranging this gives the desired condition:
\[ \frac{|T|}{n} > \frac{1}{4} \big(d - \frac{1}{4} \big)^{-1} \geq \frac{1}{5d}  .     \]
\end{proof}

We trivially have that for each $j \in [n]$, one of $\pm u_{j}$ has non-negative margin $\langle w, u_{j} \rangle_{Q} \geq 0$. Therefore $\exists u \in \{\pm u_{j} \}_{j \in [n]}$ with large margin $\frac{\langle w, u \rangle_{Q}}{\|w\|_{Q} \|u\|_{Q}} \geq \frac{1}{\sqrt{4d}}$. Therefore, we can run $2n$ instances of perceptron in parallel using each $\pm u_{j}$ as a separate initialization, and the fastest one will converge in strongly polynomial number of iterations by the margin lower bound.  

Therefore, we can use our improved frame scaling algorithm in \cref{t:mainAlgAnalysis} to improve the runtime of the strongly polynomial improper halfspace learner given in \cite{DKT}. In particular, we note that our implicit representation of the left and right scaling for approximate radial isotropic position can still be used with the perceptron algorithm to learn large margin points.

\subsection{Point Location \cite{PointLocation}}

In this subsection, we discuss the application of frame scaling to point location problem studied in \cite{PointLocation}. Here the input is a fixed set of $n$ points and we are given query access to an unknown halfspace in the form $x \to \text{sign}(\langle w, x \rangle)$. The goal is to learn the label of all $n$ points using as few queries as possible.

In the following, we discuss how frame scaling is used for this problem in \cite{PointLocation}. The first observation is quite similar to that of \cite{DKT}: if the points are in approximate $(I_{d},c)$-position according to \cref{d:introForsterTransform}, then there is a simple way to learn the label of a large fraction of points. The technical statement is as follows: 

\begin{lemma} [Lemma 4.6 in \cite{PointLocation} (Informal)]
    Let $S$ be a set of $n$ labeled examples $(x,y) \in \R{d} \times \{\pm 1\}$, such that there is an unknown linear halfspace satisfying $\forall (x,y) \in S: y = \text{sign}(\langle w, x \rangle)$. 
    Further, assume the input is in approximate $(I_{d}, c)$-position according to \cref{d:introForsterTransform}. 
    Then there is an algorithm ``IsoLearn" with the guarantee: for some integer $k \in [d]$, the algorithm makes $\tilde{O}(k)$ queries and correctly classifies a set of points $T \subseteq [n]$ with weight $\langle c, 1_{T} \rangle \geq k$.  
\end{lemma}

In other words, we can learn a large weighted fraction of points according to the weighting $c$. In order to learn all the labels, this is combined with a boosting procedure, which iteratively increases the weighting of points whose label we have not yet learned (described in Algorithm 5 in \cite{PointLocation}). Therefore, it is important that we can compute frame scalings for arbitrary weighting $c \in \R^{n}_{++}$ as given in \cref{d:introForsterTransform}. 

Using our improved frame scaling result in \cref{t:mainAlgAnalysis}, we are able to give a strongly polynomial time procedure that can be used to place points in approximate $(I_{d},c)$-position or give a certificate of infeasibility. Recall that we keep track of the square of the right scaling and leave the left scaling is implicit. Following the ideas in the previous subsection, this is sufficient if we once again switch to an appropriate inner product that depends on our scalings. 

There are a few steps in the learning algorithms given in this work that are not strictly strongly polynomial, as the main goal is to give optimal query complexity for this problem. In particular, we are required to query the halfspace on linear combinations of our points in approximate $(I_{d},c)$-position. We can once again simulate the left scaling by changing the inner product $\langle x, x' \rangle-{Q} := x^{\Tr} (U Z U^{\Tr})^{-1} x'$; but again for linear combinations we will require square root operations. Further, a key technique in this work is the notion of margin oracle, which requires queries on random Gaussian linear combinations of points. Both of these are not strongly polynomial, and we leave to future work the question of whether these issues can be avoided.

\bibliographystyle{plain}
\bibliography{refs.bib}

\end{document}